\documentclass[11pt]{amsart}

\usepackage[utf8]{inputenc}

\usepackage{amsmath, latexsym, amsfonts, amssymb, amsthm}
\usepackage{graphicx, color,hyperref,epsfig,caption,subfig}
\usepackage{pstricks}
\usepackage{mathrsfs}
\usepackage{verbatim}

\hypersetup{
    linktoc=page,
    linkcolor=red,          % color of internal links
    citecolor=blue,        % color of links to bibliography
    filecolor=blue,      % color of file links
    urlcolor=cyan,
    colorlinks=true           % color of external links
}

\numberwithin{equation}{section}

\newtheorem{theorem}{Theorem}[section]
\newtheorem{lemma}[theorem]{Lemma}
\newtheorem{proposition}[theorem]{Proposition}
\newtheorem{corollary}[theorem]{Corollary}

\newtheorem{remark}[theorem]{Remark}
\newtheorem{definition}[theorem]{Definition}

\theoremstyle{definition}

\newcommand{\C}{\mathbb{C}}
\newcommand{\R}{\mathbb{R}}

\newcommand{\N}{\mathbb{N}}

\newcommand{\E}{\mathbb{E}}
\newcommand{\D}{\mathbb{D}}

\newcommand{\hf}{\frac{_1}{^2}}

\def\eps{\varepsilon}

\newcommand{\bz}{\textbf{z}}
\newcommand{\bx}{\textbf{x}}
\newcommand{\by}{\textbf{y}}

\newcommand{\Vg}[2]{V_{#1}^{#2}}
\newcommand{\Vgd}[2]{V_{{#1},{\delta}}^{#2}}
\newcommand{\opn}{\operatorname}

\newcommand{\nug}[2]{\nu_{{#1},{#2}}^\delta}

\newcommand{\gen}{\mathbf{g}}

\newcommand{\met}{\opn{Met}(\Sigma)}
\newcommand{\dif}{\mathcal{D}(\Sigma)}
\newcommand{\meth}{\opn{Met}_{\opn{H}}(\Sigma)}
\newcommand{\modu}{\mathcal{M}(\Sigma)}
\newcommand{\sect}[2]{\Gamma^{\opn{#1}}_{#2}(S T_2 \Sigma)}
\newcommand{\sectt}[2]{\Gamma^{\opn{#1}}_{#2}(S T^2 \Sigma)}
\newcommand{\sym}{\sect{}{}}

\newcommand{\trtr}{T_g^{\opn{tt}}\met}

\begin{document}

\title[Stress-Energy in LCFT]{Stress-Energy in Liouville Conformal Field Theory on Compact Riemann Surfaces}

\author[Joona Oikarinen]{Joona Oikarinen}
\address{University of Helsinki, Department of Mathematics and Statistics,
         P.O. Box 68 , FIN-00014 University of Helsinki, Finland}
\email{ joona.oikarinen@helsinki.fi}

\begin{abstract}
We derive the conformal Ward identities for the correlation functions of the Stress--Energy tensor in probabilistic Liouville Conformal Field Theory on compact Riemann surfaces by varying the correlation functions with respect to the background metric. The conformal Ward identities show that the correlation functions of the Stress--Energy tensor can be expressed as a differential operator with meromorphic coefficient acting on the correlation functions of the primary fields of Liouville Conformal Field Theory.

Variations of the metric come in three different forms: reparametrizations, conformal scalings and deformations of the conformal structure. Conformal symmetry makes it easy to treat variations of the metric that do not deform the conformal structure. Variations  that deform the conformal structure have to be treated separately, and this part of the computation relies on regularity and integrability properties of the correlation functions of Liouville Conformal Field Theory.
\end{abstract}

\maketitle
	
	\tableofcontents
	
\section{Introduction and Main Result}

\noindent Conformal field theories (CFTs) are field theories with conformal symmetry. The study of two-dimensional CFTs was initiated by Belavin, Polyakov and Zamolodchikov (BPZ) in \cite{BPZ}, where the concept of local conformal symmetry was emphasized. They showed that, in the context of two-dimensional Euclidean quantum field theory, the Lie algebra corresponding to conformal symmetry has infinitely many generators, and this infinite-dimensional symmetry leads to strong constrains on the theory. This led BPZ to formulate the conformal bootstrap hypothesis, which produces explicit expressions for the correlation functions of CFTs. The mathematical formulation of local conformal symmetry and the BPZ approach to CFT has been a goal for mathematicians ever since.

The notion of infinite-dimensional conformal symmetry in two-dimensional Euclidean space is rather elusive, since in this setting there exists no corresponding infinite-dimensional symmetry group \cite{schottenloher}. However, global conformal symmetry for diffeomorphism covariant field theories still makes sense in terms of a property called the 
\emph{Weyl invariance}. This gives a notion of global conformal symmetry in any dimension. The two-dimensional case stands out when one formulates a local version of this symmetry. What happens is that the object that encodes symmetries in field theories, (a component of) the \emph{Stress-Energy Tensor} (SE-tensor), is holomorphic in two-dimensional Weyl invariant field theories. Moreover, in the presence of the \emph{primary fields}, the SE-tensor picks up poles at the locations of the primary fields, leading to a meromorphic SE-tensor. This meromorphicity is the reason why one expects projective unitary representation of the Witt algebra, which are equivalent to unitary representations of the Virasoro algebra, to appear in the Hilbert space picture of CFT.

The picture is slightly complicated by a phenomenon called \emph{Weyl anomaly}, which usually arises when quantizing a Weyl invariant theory. The Weyl anomaly describes the breaking of Weyl invariance. The resulting effect for the SE-tensor is that it fails to be holomorphic in regions where the underlying space has curvature. The Weyl anomaly is also how the \emph{central charge}, an important parameter in CFTs, appears, as the measure of the strength of the Weyl anomaly. In the Hilbert space picture for CFTs, the central charge parametrizes the different central extensions of the symmetry algebra corresponding to local conformal symmetry, the Witt algebra. The central extension is called the Virasoro algebra.

Classical Liouville theory and its conformal symmetry have appeared in the mathematics literature already in the works of Poincaré and Picard in relation to uniformization of punctured Riemann surfaces \cite{unif}. Even though they did not explicitly talk about conformal symmetry, the classical SE-tensor appeared and played a crucial role in their work. Especially, they discovered the classical version of the so-called conformal Ward identities, but they were not able to derive explicit formulae for certain coefficients\footnote{These coefficients are called the \emph{accessory parameters} in complex geometry.} appearing in these identities. It turns out that discovering the correct formula is almost automatic from the point of view of CFT. The main result of the present work is the derivation of the conformal Ward identities in the quantized setting, starting from the probabilistic construction of the path integral of Liouville CFT, and then differentiating it with respect to the background metric.

The rigorous study of Euclidean Liouville quantum field theory was initiated in \cite{DKRV}, and since then there has been a lot of progress in understanding the conformal bootstrap framework in the case of probabilistic Liouville CFT, which is based on a probabilistic construction of the path integral. The path integral can be studied by means of Gaussian Multiplicative Chaos, which in this case is a random measure constructed using the exponential of the Gaussian Free Field. 

Before describing the main result and the rigorous probabilistic constructions of the relevant objects, we describe the general picture of conformal symmetry in two-dimensional Euclidean field theory, giving extra attention to Liouville field theory.

\subsection{Global conformal symmetry}

We start by giving a short description of conformal symmetry on a classical level. Let $(\Sigma,g)$ be an orientable compact  Riemannian surface. We say that a map $\psi: \Sigma \to \Sigma$ is conformal if it is a diffeomorphism and  $\psi^* g = e^\omega g$ for some $\omega \in C^\infty(\Sigma,\R)$, where $\psi^*$ denotes the pull-back and $C^\infty(\Sigma,\R)$ denotes the space of smooth real-valued functions on $\Sigma$. We denote the space of conformal maps on $(\Sigma,g)$ by $\opn{Conf}(\Sigma,g)$.

Consider a classical field theory on $(\Sigma,g)$ described by an action functional $S$. Maybe the simplest way to try to formulate global conformal symmetry would be to require that $S$ is \emph{conformally invariant}
\begin{align}\label{S_conformal_symmetry}
S(\varphi \circ \psi,g) = S(\varphi,g)\,, \quad \psi \in \opn{Conf}(\Sigma,g)\,,
\end{align}
for all field configurations $\varphi$ and conformal maps $\psi$. The symmetry group of $S$ is the biggest subgroup of $\opn{Diff}(\Sigma)$ for which \eqref{S_conformal_symmetry} holds. If one were to capture the full idea of infinite-dimensional symmetry envisioned by BPZ, one would expect the symmetry group of an action of a CFT to be bigger than just $\opn{Conf}(\Sigma,g)$, which is always finite-dimensional on Riemannian surfaces (sometimes even empty) \cite{schottenloher}, and it seems unlikely that this would lead to an infinite-dimensional local conformal symmetry. Because of this, we formulate conformal symmetry in a different way, which we describe next.

\subsubsection{Weyl invariance} Let $\opn{Diff}(\Sigma)$ be the set of all diffeomorphisms on the surface $\Sigma$. If we assume that the action functional $S$ is \emph{diffeomorphism covariant}
\begin{align}\label{S_diffeo_covariance}
S(\varphi \circ \xi^{-1}, g) = S(\varphi, \xi^* g)\,, \quad \xi \in \opn{Diff}(\Sigma)\,,
\end{align}
 then we can generalize the conformal invariance property \eqref{S_conformal_symmetry}. Indeed, if we apply the inverse of a conformal map $\psi$ to a diffeomorphism covariant action, we get
\begin{align}
S(\varphi \circ \psi^{-1},g) = S(\varphi, \psi^* g) = S(\varphi, e^\omega g)\,,
\end{align}
for some $\omega \in C^\infty(\Sigma,\R)$. Thus, if the action functional is \emph{Weyl invariant}
\begin{align}\label{S_weyl_invariance}
S(\varphi,e^\omega g) = S(\varphi,g)\,, \quad \omega \in C^\infty(\Sigma,\R)\,,
\end{align}
then \eqref{S_conformal_symmetry} holds. The upshot is that \eqref{S_weyl_invariance} makes sense even when no conformal maps exist on $(\Sigma,g)$. This motivates us to take the two properties \eqref{S_diffeo_covariance} and \eqref{S_weyl_invariance} as the definition of global conformal symmetry in classical field theory. Note that this definition does not make it clear what is the corresponding symmetry group (as a subgroup of $\opn{Diff}(\Sigma)$), but it allows us to formulate a fruitful notion of local conformal symmetry, which we  describe next.

\subsection{Local conformal symmetry}

Next we sketch how diffeomorphism covariance and Weyl invariance together lead to a description of local conformal symmetry on the classical level, which is the first step in understanding the conformal symmetry described by BPZ. Diffeomorphism covariance \eqref{S_diffeo_covariance} says that symmetries of $S$ are encoded in its $g$-dependency\footnote{By symmetry we mean that the value of $S$ does not change when we apply some diffeomorphism to $\varphi$. The diffeomorphism covariance \eqref{S_diffeo_covariance} allows us to move the action of this diffeomorphism to the metric, and in this sense symmetries of $S$ are encoded in its $g$-dependency.}. It is then natural to describe local symmetries in terms of the \emph{classical stress-energy tensor}
\begin{align}
T_{\alpha \beta}(z) &:= - 4 \pi \frac{\delta S(\varphi_\star,g)}{\delta g^{\alpha \beta}(z)}\,,
\end{align}
where $(g^{\alpha \beta})_{\alpha,\beta}$ is the inverse matrix of $(g_{\alpha \beta})_{\alpha,\beta}$, and  $\varphi_\star$ is the stationary value of the action $S$, i.e. the classical solution. We expect that diffeomorphism covariance and Weyl invariance translate into special properties for the SE-tensor. Indeed, if we take a flow of diffeomorphisms $(\psi_\eps)_{\eps \geq 0}$ generated by a smooth vector field $u$, then diffeomorphisms covariance implies
\begin{align*}
0 &= \partial_\eps|_0S(\varphi_\star \circ \psi_\eps,\psi_\eps^* g) = \partial_\eps|_0 S(\varphi_\star \circ \psi_\eps,g) + \partial_\eps|_0 S(\varphi_\star,\psi_\eps^* g) = \partial_\eps|_0 S(\varphi_\star,\psi_\eps^* g)\,,
\end{align*}
where we used the fact that $\partial_\eps|_0 S(\varphi_\star \circ \psi_\eps,g)=0$ because $\varphi_\star$ is the stationary value of $S$. Using the fact that the flow $(\psi_\eps)_{\eps \geq 0}$ is generated by $u$, we get that
\begin{align*}
(\psi_\eps^* g)^{\alpha \beta} &=  g^{\alpha \beta} + \eps(\nabla^\alpha u^\beta + \nabla^\beta u^\alpha) + \mathcal{O}(\eps^2)\,,
\end{align*}
where $\nabla$ is the gradient and $g^{\alpha \beta}$ denotes components of the inverse of $g$. It follows from the definition of the functional derivative (or formally, the chain rule), that
\begin{align*}
\partial_\eps|_0 S(\varphi_\star,\psi_\eps^* g) = \int_\Sigma \partial_\eps|_0 (\psi_\eps^* g)^{\alpha \beta}(z) \tfrac{\delta S(\varphi_\star,g)}{\delta g^{\alpha \beta}(z)} dv_g(z) = \int_\Sigma(\nabla^\alpha u^\beta(z) + \nabla^\beta u^\alpha(z))T_{\alpha \beta}(z) dv_g(z) &= 0 \,,
\end{align*}
where we use the Einstein summation convention. By integrating by parts and using the fact that $u$ is arbitrary, we get
\begin{align}\label{continuity_eq}
\nabla^\alpha T_{\alpha \beta} = 0\,.
\end{align}
The above equation is the local version of diffeomorphism covariance \eqref{S_diffeo_covariance}.

Next, to probe the implications of Weyl invariance, we take a family of smooth real-valued functions $(\omega_\eps)_{\eps \geq 0}$ with $\omega_0=0$. Then we have
\begin{align*}
 (e^{\omega_\eps} g)^{\alpha \beta} &= (1 + \eps \dot \omega) g^{\alpha \beta} + \mathcal{O}(\eps^2)\,,
\end{align*}
where $\dot \omega= \partial_\eps|_0 \omega_\eps$. Now, Weyl invariance leads to
\begin{align*}
\partial_\eps|_0 S(\varphi, e^{\omega_\eps} g) &= \int_\Sigma \partial_\eps|_0 (e^{\omega_\eps} g)^{\alpha \beta}(z) \tfrac{\delta S(\varphi,g)}{\delta g^{\alpha \beta}(z)} dv_g(z) = \frac{1}{4\pi} \int_\Sigma \dot \omega(z) g^{\alpha \beta}(z) T_{\alpha \beta}(z) dv_g(z)\,.
\end{align*}
Because $\dot \omega$ is an arbitrary smooth function, we get
\begin{align}\label{traceg}
\opn{tr}_g(T) &= g^{\alpha \beta} T_{\alpha \beta} =  0\,,
\end{align}
which is a local version of Weyl anomaly.

Two-dimensional CFT is often studied in conformal coordinates. Given a Riemannian metric, we call a conformal atlas (in two-dimensions) an atlas of charts where the metric takes the form
\begin{align*}
g(z) &= e^{\sigma(z)} |dz|^2 = \hf e^{\sigma(z)}(dz \otimes d \bar z + d \bar z \otimes dz)\,.
\end{align*}
A collection of such charts yields a complex structure, meaning that the transition maps between charts in a conformal atlas turn out to be holomorphic, see e.g. Theorem 3.11.1 in \cite{Jost}. A maximal conformal atlas is called a conformal structure.

It can then be checked that in a conformal coordinate $z$, the combination of the equations \eqref{continuity_eq} and \eqref{traceg} leads to the holomorphicity condition
\begin{align}\label{T_holom}
\partial_{\bar z} T_{zz} = 0\,.
\end{align}
This conclusion is special for the two-dimensional case, and could be seen as the main reason for the richness of local conformal symmetry in two-dimensional space. Indeed, one could now expand
\begin{align}\label{T_series}
T_{zz}(z) &= \sum_{n=0}^\infty \ell_n z^n\,,
\end{align}
and in the Hilbert space picture (see e.g. \cite{Kup}) the complex coefficients $\ell_n$ would become operators on the Hilbert space of the quantized theory.

\subsection{Conformal symmetry in Euclidean quantum field theory}
In the context of Euclidean quantum field theory\footnote{By this we mean that the underlying space is Euclidean (or Riemannian), so there is no time dimension.}, for a given action functional $S$, expected values of observables are formally given in terms of a \emph{path integral}
\begin{align}\label{path_integral}
\langle F \rangle_g &:= \int F(\varphi) e^{-S(\varphi,g)} \text{d} \varphi\,,
\end{align}
where $\text{d} \varphi$ is a formal Lebesgue measure on the space of field configurations. The above should be viewed as the integral of the functional $F$ with respect to the measure $e^{-S(\varphi,g)} \text{d} \varphi$ on some space of functions. In the context of CFT we then expect \eqref{path_integral} to satisfy \emph{diffeomorphism covariance}
\begin{align}
\langle F \rangle_{\Sigma,\psi^* g} &= \langle \psi_* F \rangle_{\Sigma,g}\,,  \quad \psi \in \opn{Diff}(\Sigma)\,,
\end{align}
where $(\psi_* F)(\varphi) := F(\varphi \circ \psi)$, and \emph{Weyl anomaly}
\begin{align}\label{F_axioms}
\langle F \rangle_{\Sigma,e^\omega g} &= e^{c A(\omega,g)} \langle F \rangle_{\Sigma,g}\,, \quad \omega \in C^\infty(\Sigma,\R)\,,
\end{align}
where $c$ is a constant called the \emph{central charge}, and  $A$ is a functional describing the conformal anomaly, given by 
\begin{align}\label{anomaly}
A(\varphi,g) &= \frac{1}{48 \pi} \int_\Sigma ( \tfrac{1}{2} |d \varphi|_g^2 +  K_g \varphi) dv_g\,,
\end{align}
where we denoted by $d$ the exterior derivative, by $|\cdot|_g$ the norm induced by $g$, by $K_g$ the scalar curvature, and by $dv_g$ the volume form of the metric $g$. Heuristically, the exact Weyl invariance present in the classical setting \eqref{S_weyl_invariance} does not appear in the path integral setting due to the singular nature of the measure $\text{d}{\varphi}$, which breaks the Weyl invariance \cite{MaMi}.

In CFT the interest is on the \emph{correlation functions} of \emph{primary fields} $V^\alpha(x)$
\begin{align}\label{intro1}
\langle \prod_{j=1}^N V^{\alpha_j}(x_j) \rangle_{\Sigma,g}
\end{align}
The defining properties for primary fields, in the case of scalar fields, is that we have the diffeomorphism covariance
\begin{align}\label{intro_diffeo_weyl}
\langle \prod_{j=1}^N V^{\alpha_j}(x_j) \rangle_{\Sigma,\psi^* g} &= \langle \prod_{j=1}^N V^{\alpha_j}(\psi(x_j)) \rangle_{\Sigma,g}\,,
\end{align}
and the Weyl anomaly
\begin{align}\label{intro_weyl}
\langle \prod_{j=1}^N V^{\alpha_j}(x_j) \rangle_{\Sigma,e^\omega g} &= e^{cA(\omega,g) -\sum_{j=1}^N \Delta_{\alpha_j} \omega(x_j)} \langle \prod_{j=1}^N V^{\alpha_j}(x_j) \rangle_{\Sigma,g}\,,
\end{align}
where the constant $\Delta_\alpha$ is the \emph{conformal weight} of the primary field $V^\alpha$. The difference between the two Weyl anomalies \eqref{F_axioms} and \eqref{intro_weyl} is due to the singular nature of the fields $V^\alpha$.

In this setting, the SE-tensor is formally defined by varying \eqref{path_integral} with respect to the inverse of the metric
\begin{align*}
\langle T_{\mu \nu}(z) F \rangle_{\Sigma,g} &:= 4 \pi \tfrac{\delta}{\delta g^{\mu \nu} (z)} \langle F \rangle_{\Sigma,g} = -4\pi \int\tfrac{\delta S(\varphi,g)}{\delta g^{\mu \nu}(z)} F(\varphi) e^{-S(\varphi,g)} \text{d} \varphi\,.
\end{align*}
What we will be interested in are the higher variations with respect to $g$ of the correlation functions
\begin{align}\label{introdg}
\langle \prod_{k=1}^n T_{\mu_k \nu_k}(z_k) \prod_{j=1}^N V^{\alpha_j}(x_j) \rangle_{\Sigma,g} &:= (4\pi)^n \prod_{k=1}^n \tfrac{\delta}{\delta g^{\mu_k \nu_k}(z_k)} \langle \prod_{j=1}^N V^{\alpha_j}(x_j) \rangle_{\Sigma,g}\,,
\end{align}
where $g^{\mu \nu}$ denotes the components of the inverse of $g$, that is, $g^{\alpha \beta} g_{\beta \gamma} = \delta^\alpha_\gamma$. The precise definition of the SE-tensor will be established later in Proposition \ref{T_pointwise}. The symmetries \eqref{intro_diffeo_weyl} and \eqref{intro_weyl} strongly constrain the correlation functions \eqref{introdg} of the SE-tensor, which can be seen from the fact that an arbitrary smooth metric $g$ on $\Sigma$ can be written in the form
\begin{align}\label{intro2}
g &= e^\omega \psi^* g_\tau\,,
\end{align}
where $(g_\tau)_{\tau \in \modu}$ is a collection of distinguished metrics, parametrized by the \emph{moduli space} $\modu$ of the surface $\Sigma$. The moduli space can be seen as the space of different conformal (or complex) structures of the surface, and in view of Equations \eqref{intro_diffeo_weyl} and \eqref{intro_weyl}, the correlation functions \eqref{intro1} depend non-trivially only on the conformal structure generated by $g$, i.e. on the metric $g_\tau$ appearing in \eqref{intro2}. Since the dependence on $g_\tau$ of the correlation functions is not determined by the symmetries \eqref{intro_diffeo_weyl} and \eqref{intro_weyl}, it has to be studied separately in each CFT. For compact surfaces, $\modu$ is finite-dimensional, and thus conformal symmetry determines the behaviour of the SE-tensor up to a finite dimensional degree of freedom of the metric.

Working on the plane (and hence with trivial moduli space and no curvature), BPZ argued with theoretical physics level of rigour, that in conformal coordinates the correlation functions of the SE-tensor are meromorphic in the $z_k$ variables and can be written as a differential operator acting on correlation functions of the primary fields in the $x_k$ variables. The resulting identities are called the \emph{conformal Ward identities}, and they completely describe the poles of the correlation function.

Weyl invariance implies that the SE-tensor is traceless, and in flat regions of space this holds even in the presence of Weyl anomaly. This leaves two non-trivial components for the SE-tensor, since symmetric tensors have at most three non-trivial components. The two non-trivial components of the SE-tensor then are
\begin{align}
T_{zz} &= \tfrac{1}{4}(T_{11} - T_{22} - 2 i T_{12})\,, \\
T_{\bar z \bar z} &= \tfrac{1}{4}(T_{11}-T_{22} + 2 i T_{12})\,,
\end{align}
where $z = x_1+ix_2$ and the components on the right-hand side denote the components of $T$ in the coordinate $(x_1,x_2)$. Later, Eguchi and Ooguri in \cite{EgOo} studied the SE-tensor of CFT on arbitrary compact Riemann surfaces, including the variations of the moduli $\tau \in \modu$ in their computations. They showed, with the theoretical physics level of rigour, that also in the presence of modular variations the correlation functions of the SE-tensor remain meromorphic and single-valued, and they satisfy the conformal Ward identities.

\subsection{Conformal symmetry in classical Liouville field theory}

Classical Liouville field theory originally appeared in the context of uniformization of punctured Riemann surfaces. The punctures are located at points $x_i \in \Sigma$ and to each puncture there is an associated weight $\chi_i \in \R$. For $y=(y_1,\hdots,y_n) \in \Sigma^n$ and $\chi=(\chi_1,\hdots,\chi_n) \in \R^n$ the action functional, roughly speaking, is given by
\begin{align}\label{punctured action}
S_L(\varphi,g,\chi,y) &= \int_\Sigma \big( \hf |d\varphi|_g^2 + K_g \varphi + 2 e^\varphi \big) dv_g  -2\pi \sum_{i=1}^n \chi_i (\varphi(y_i) + \sigma(y_i) )\,,
\end{align}
where $g=e^\sigma \delta$ in a fixed conformal coordinate chart $\C$. In reality the above action has to be regularized and the fields over which it is minimized is a certain space of functions that have logarithmic singularities at the points $y_i$ \cite{LRV}. The minimizer $\varphi_\star^{(\chi,y)}$ of this action has the property that the metric $e^{\varphi_\star^{(\chi,y)}}g$ has a constant negative curvature $-2$ outside of the punctures $y_i$, and the metric has conical singularities at the punctures.

The action functional has the following diffeomorphism covariance and anomalous Weyl invariance (Weyl anomaly) properties
\begin{align}\label{S_L_symmetry}
S_L(\psi \circ \varphi, \psi^*g,\chi,y) &= S_L(\varphi,g,\chi, \psi(y))\,, \quad \psi \in \opn{Diff}(\Sigma)\,, \nonumber \\
S_L(\varphi,e^\omega g,\chi,y) &= S_L(\varphi+\omega,g,\chi,y) - A(\omega,g) +  \sum_{i=1}^n  \chi_i(1- \tfrac{\chi_i}{4} ) \omega(y_i) \,, \quad \omega \in C^\infty(\Sigma,\R)\,,
\end{align}
where we used the abuse of notation $\psi(y) = (\psi(y_1),\hdots,\psi(y_n))$. The form of the action functional implies that the field $\varphi_\star^{(\chi,y)}$ has logarithmic singularities at the points $y_i$. This has the consequence that the SE-tensor will no longer be holomorphic, but will pick up first and second order poles at these points. The formula describing these poles is the classical conformal Ward identity, which in a slightly simplified form\footnote{The formula holds in regions where $g$ is Euclidean, i.e. $\sigma=0$.} is given by
\begin{align}\label{i_T}
T_{zz}(z) &=  \sum_{i=1}^n \Big( \frac{\chi_i(1-  \frac{\chi_i}{4}) }{(z-y_i)^2} - \frac{\partial_{y_i} S_L(\varphi_\star^{(\chi,y)},g,\chi,y)}{z-y_i} \Big)\,,
\end{align}
The form of this identity directly follows from the symmetries \eqref{S_L_symmetry}, the second order pole originating from the Weyl anomaly, and the first order pole originating from the diffeomorphism covariance, which can be quickly seen by the following heuristic. Fix a conformal coordinate $z$ and consider an infinitesimal perturbation of the metric $g$, given by
\begin{align*}
(g_\eps)^{zz} &= g^{zz} + \eps \delta_{g,z}
\end{align*}
where $\eps$ is infinitesimal and $\delta_{g,z}$ is the delta function at $z$. For simplicity we work on the sphere $\Sigma = \mathbb{S}^2$ and we assume that the metric is Euclidean around the points $z$ and $y_i$, i.e. $g(w)=e^{\sigma(w)} |dw|^2$ and  $\sigma(w) = 0$ for $w$ close enough to $z$ or $y_i$. Any metric $\hat g$ on the sphere can be written in the form $\hat g = e^\omega \psi^* g$. Thus, our perturbed metric also admits the form $e^{\omega_\eps}\psi_\eps^* g$. One can then check that in a conformal coordinate $w$ the diffeomorphism $\psi_\eps$ satisfies the partial differential equation
\begin{align*}
\partial_{\bar w} \psi_\eps(w) &= - \tfrac{\eps}{4} \delta_{g,z}(w) e^{\sigma(w)} + \mathcal{O}(\eps^2)\,.
\end{align*}
The partial derivative $\partial_{\bar w}$ is inverted by the Cauchy transform, leading to the Cauchy kernel
\begin{align}\label{i_psi}
\psi_\eps(w) &= - \frac{\eps}{4\pi} \frac{1}{w-z} + \mathcal{O}(\eps^2)\,.
\end{align}
It can also be checked that 
\begin{align}\label{i_phi}
\omega_\eps(w) &= \frac{\eps }{4\pi} \Big( \frac{1}{(w-z)^2} + \frac{\partial_w \sigma(w)}{w-z} \Big) + \mathcal{O}(\eps^2)\,.
\end{align}
Now the Ward identity can be derived by using the symmetries of the action functional as follows
\begin{align*}
T_{zz}(z) &= 4 \pi  \partial_\eps|_0 S_L(\varphi_\star^{(\chi,y)},e^{\omega_\eps} \psi_\eps^* g,\chi, y) \\
&= 4\pi \partial_\eps|_0 \big( \sum_{i=1}^n \tfrac{1}{2}(\chi_i- \tfrac{\chi_i^2}{4}) \omega_\eps( \psi_\eps(y_i)) - A(\omega_\eps \circ \psi_\eps , g) +  S_L(\varphi_\star^{(\chi,\psi_\eps(y))} \circ \psi_\eps  + \omega_\eps \circ \psi_\eps,  g,\chi, \psi_\eps(y))   \big)\,.
\end{align*}
By using the fact that $\varphi_\star^{(\chi,y)}$ is the stationary value and that $\sigma(z)=0$, the above derivative simplifies to
\begin{align*}
T_{zz}(x) &= 4 \pi \partial_\eps|_0 \big(  \sum_{i=1}^n \chi_i(1 - \frac{\chi_i}{4}) \omega_\eps(y_i) + S_L(\varphi_\star^{(\chi,\psi_\eps(y))} \circ \psi_\eps,g,\chi, \psi_\eps(y) )   \big) \\
&=    4 \pi \sum_{i=1}^n \chi_i(1 - \tfrac{\chi_i}{4}) \partial_\eps|_0 \omega_\eps(y_i) + 4\pi \sum_{i=1}^n \partial_\eps \psi_\eps(y_i) \partial_{y_i} S_L(\varphi_\star^{(\chi,y)},g,\chi,y) \,.
\end{align*}
Now \eqref{i_T} follows from Equations \eqref{i_psi} and \eqref{i_phi}.

\subsection{Path integrals and Liouville Conformal Field Theory}

The constructive approach to quantum field theory aims to construct correlation functions (and other observables) of fields as concrete integrals over an infinite-dimensional space of field configurations. In this setting, for a functional $F$ on the field configuration space, we consider an integral of the form
\begin{align}\label{intro_pi}
\langle F \rangle_{\Sigma,g} &= \int F(\varphi) e^{-S(\varphi,g)} \text{d}\varphi\,,
\end{align}
where $S$ is the action functional describing the specific theory at hand, and $\text{d} \varphi$ is a tentative uniform measure on the space of field configurations. Then properties like \eqref{intro_diffeo_weyl} should be derivable from the properties of the path integral, once it is rigorously defined.

The action functional of Liouville field theory is 
\begin{align*}
S_L(\varphi,g) &= \frac{1}{4\pi} \int_\Sigma \big( |d \varphi|_g^2 + Q K_g \varphi + 4 \pi \mu e^{\gamma \varphi} \big) dv_g\,,
\end{align*}
where $\gamma \in (0,2), \mu \in (0,\infty)$ and $Q = \frac{2}{\gamma} + \frac{\gamma}{2}$. Liouville theory turns out to be a CFT with the primary fields given by
\begin{align}\label{introliou}
V_\alpha(x) &= e^{\alpha \varphi(x)}
\end{align}
for $\alpha \in \C$. The corresponding conformal weight is
\begin{align}
\Delta_\alpha &= \tfrac{\alpha}{2} (Q-\tfrac{\alpha}{2})\,.
\end{align}
A formal computation in \eqref{introdg} shows that in conformal coordinates, where $g=e^{\sigma}|dz|^2$, the SE-tensor should be given by the field
\begin{align}\label{intro_sefield}
T_{zz}(z) &= Q \partial_z^2 \big(\varphi(z) + \tfrac{Q}{2} \sigma(z)\big) - \big(\partial_z (\varphi(z) + \tfrac{Q}{2} \sigma(z) \big)^2\,,
\end{align}
where the $(\partial_z \varphi)^2$-term is supposed to be ``Wick ordered".

A rigorous construction of the path integral for the Liouville theory was carried out in \cite{DKRV} on the Riemann sphere, and later on other compact Riemann surfaces in \cite{DRV, GRV16}. The correlation functions of the fields \eqref{introliou} have a concrete path integral construction for certain real values of $\alpha$. The constrains on the possible values of $\alpha$ are called the \emph{Seiberg bounds}, to be specified later. The conformal Ward identities on the Riemann sphere were derived in \cite{ward} using the field definition of the SE-tensor \eqref{intro_sefield} and Gaussian integration by parts. In \cite{KO} the conformal Ward identities on the Riemann sphere were derived by varying the metric.

The SE-tensor of a CFT is closely related to its Virasoro algebra, and the conformal Ward identities have algebraic consequences. From the path integral of LCFT it is possible to construct the physical Hilbert space, see \cite{Kup}. It then should be possible to construct a unitary representation of the Virasoro algebra on this Hilbert space using the conformal Ward identities, see Section 4 of \cite{KO} for discussion. The Virasoro representation for Liouville CFT was recently rigorously constructed by a different method in \cite{BGKRV}.

In this paper we use the approach of \cite{KO} to derive the conformal Ward identities on compact surfaces with genus at least two. This approach is more natural from geometric point of view, as was explained in the previous sections, and does not rely on the ansatz \eqref{intro_sefield} obtained from formal computations\footnote{The origin of the ansatz is explained in Appendix \ref{appendix_b}.}. The presence of a non-trivial moduli space complicates the analysis, since a separate proof, showing that the modular variations of the correlation functions are well-defined, is required. The proof relies on properties of the path integral of Liouville theory, and as such does not say anything general about modular variations in CFTs defined by a path integral. If the metric is varied in a way that leaves the conformal structure unchanged, the corresponding variation of the correlation functions can be computed by using the conformal symmetry, and thus this part of the computation works in other CFTs, too.

\subsection{Main result}

The conformal Ward identity can be conveniently expressed using the Green function $\mathcal{G}$ of the anti-holomorphic derivative $\nabla^z = g^{z \bar z} \partial_{\bar z}$, viewed as a map $\Gamma(T^1 \Sigma) \to \Gamma(T^2 \Sigma)$ (see Section \ref{section_geometry}). In conformal coordinates it satisfies
\begin{align}\label{introg}
\mathcal{G}^z_{ww}(z,w) &= \frac{1}{4\pi} \frac{1}{z-w} + \mathcal{R}(z,w) \,,
\end{align}
for some smooth function $\mathcal{R}$. Then the conformal Ward identity for a hyperbolic metric $g$ (i.e. a metric with constant negative scalar curvature) is
\begin{align}\label{introward}
&\frac{1}{4\pi} \langle \prod_{k=1}^n T_{z_k z_k}(z_k) \prod_{j=1}^N \Vg{g}{\alpha_j}(x_j) \rangle_g \\
&= - \frac{c}{12} \sum_{i=2}^n \nabla^3_{z_i} \mathcal{G}^{z_i}_{z_1 z_1}(z_i,z_1) \langle \prod_{k \neq 1,i} T_{z_k z_k}(z_k) \prod_{j=1}^N \Vg{g}{\alpha_j}(x_j) \rangle_g \nonumber \\
& \quad +  \sum_{i=2}^n \big( 2 \nabla_{z_i} \mathcal{G}^{z_i}_{z_1 z_1}(z_i,z_1) + \mathcal{G}^{z_i}_{z_1 z_1} (z_i,z_1) \nabla_{z_i} \big) \langle \prod_{k=2}^n T_{z_k z_k}(z_k) \prod_j \Vg{g}{\alpha_j}(x_j) \rangle_g \nonumber \\
& \quad  + \sum_{j=1}^N \big( \Delta_{\alpha_j} \nabla_{x_j} \mathcal{G}^{x_j}_{z_1 z_1} (x_j,z_1) + \mathcal{G}^{x_j}_{z_1 z_1}(x_j,z_1) \nabla_{x_j} \big) \langle \prod_{k=2}^n T_{z_k z_k}(z_k) \prod_j \Vg{g}{\alpha_j}(x_j) \rangle_g\,, \nonumber \\
& \quad + \frac{1}{4\pi} \langle T_m(z_1) \prod_{k=2}^n T_{z_k z_k}(z_k) \prod_j \Vg{g}{\alpha_j}(x_j) \rangle_g\,. \nonumber
\end{align}
Above $\nabla$ is the covariant derivative, and $c$ is the central charge that already appeared in \eqref{intro_weyl}. The last term on the right-hand side is the contribution coming from the deformation of the conformal structure of the Riemann surface, and its precise definition will be given in the proof of Theorem \ref{theorem}.

We remark that the function $\mathcal{R}$ in \eqref{introg} is not holomorphic, and thus from the formula above it is not immediately clear that the left-hand side is a meromorphic function. However, the non-holomorphic contributions from $\mathcal{R}$ should cancel out in the end, as was argued at the end of Section 2 in \cite{EgOo}. For non-hyperbolic metrics one also gets a non-holomorphic contribution, see Equation (46) in \cite{EgOo}. In the case of the sphere, $\mathcal{R}$ vanishes, and $\mathcal{G}$ essentially reduces to the Cauchy transform, so the Ward identities take their familiar form. Our result is the following.
\begin{theorem}
Let $(x_1,\hdots,x_N) \in \Sigma^N$ be non-coinciding points. Assume that $(\alpha_1,\hdots,\alpha_N) \in \R^N$ satisfy the Seiberg bounds. Then the LCFT correlation functions \eqref{intro1} on $(\Sigma,g)$ are smooth with respect to the metric $g$. The derivatives \eqref{introdg} exist and are smooth in $z_k$ and $x_j$ in the region of non-coinciding points. For hyperbolic metric $g$ the derivatives \eqref{introdg} satisfy the conformal Ward identity \eqref{introward}.
\end{theorem}

The structure of the article is the following. In Section \ref{section_geometry} we summarize the geometric background needed for the rest of the article. In Section \ref{section_lcft} we recall the definition of the LCFT correlation functions. In Section \ref{section_green} we compute the variations with respect to the metric of the Green function of the Laplace--Beltrami operator. In Section \ref{section_correlation} we show that the variations with respect to the metric of the LCFT correlation functions are well-defined, where the main focus is on the variations of the moduli. Finally, in Section \ref{section_ward}, we derive the conformal Ward identities.

\subsection*{Acknowledgements} We thank Antti Kupiainen for useful discussions and guidance during this project. This work is financially supported by DOMAST and Academy of Finland.

\section{Geometric background}\label{section_geometry}

Let $(\Sigma,g)$ be an orientable compact Riemann surface with genus $\gen \geq 2$. We denote the space of all Riemannian metrics on $\Sigma$ by $\met$, the space of smooth tensor fields of order $(n,m)$ by $\Gamma(T^n_m \Sigma)$, and the space of symmetric smooth tensor fields of order $(n,m)$ by $\Gamma(ST^n_m \Sigma)$. The corresponding subspaces of trace-free tensors are denoted by $\Gamma^{\opn{tf}}(T^n_m \Sigma)$ and $\Gamma^{\opn{tf}}(ST^n_m \Sigma)$. Note that $\met$ is an open subset of $\sect{}{}$, which allows us to identify the tangent space $T_g \met$ with $\sect{}{}$. We endow $T_g \met$ with the $L^2$-inner product
\begin{align}\label{tangent_ip}
(f_1, f_2)_g &= \int_\Sigma (f_1)^{\alpha \beta} (f_2)_{\alpha \beta} dv_g\,,
\end{align}
where we use the Einstein summation convention for repeated Greek indices. We denote the group of smooth diffeomorphisms on $\Sigma$ by $\dif$. It acts on $\met$ by pullback 
\begin{align}
\psi \cdot g &:=\psi^* g \,, \quad \psi \in \dif, \; g \in \met\,.
\end{align}
The group of smooth real-valued functions $C^\infty(\Sigma)$ acts on $\met$ by Weyl transformations
\begin{align}
\varphi \cdot g &:= e^\varphi g\, \quad \varphi \in C^\infty(\Sigma), \; g \in \met\,.
\end{align}
We say that a metric $g$ is hyperbolic if the scalar curvature is constant $K_g=-2$. The subspace of $\met$ consisting of hyperbolic metrics is denoted by $\meth$. It is a smooth sub manifold of $\met$. We denote the covariant derivative (or the Levi--Civita connection) of $g$ by $\nabla$ and the Lie derivative of a tensor field $f$ along a vector field $u$ by $\mathcal{L}_u f$.

We call the orbit of $g$ under the action of $C^\infty(\Sigma)$ the \emph{conformal class} of $g$. For our purposes the following classical result will be fundamental. Any metric $g \in \met$ can be uniquely written in the form
\begin{align}
g &= e^\varphi h
\end{align}
where $h \in \meth$ and $\varphi \in C^\infty(\Sigma)$. Thus, in each conformal class, there exists a unique hyperbolic metric. The moduli space $\modu$ of the surface $\Sigma$ is given by the quotient
\begin{align}
\modu &:= \big(\met / C^\infty(\Sigma)\big)/\dif = \meth/\dif \,.
\end{align}
The moduli space is an orbifold of complex dimension $3\gen-3$. It fails to be a manifold because the action of $\mathcal{D}(\Sigma)$ on $\meth$ has fixed points, and thus the quotient creates orbifold points, corresponding to the fixed points. This means that the orbifold points are (equivalence classes of) hyperbolic metrics that admit non-trivial isometries. For a reference on moduli spaces of compact Riemann surfaces, see e.g. Chapter 12 of \cite{FaMa} or Chapter 12.5.1 of \cite{handbook} and references therein.

\begin{lemma}
	Let $f \in \sym$ and define $g_\eps = g + \eps f$. Then
	%we will perturb $g^{-1}$ instead, so careful with minus signs! (*)
	\begin{align}\label{K_var}
	\partial_{\eps}|_0 K_{g_\eps}(z) &= - \hf K_g \opn{tr}_g (f) + \nabla^{\alpha}\nabla^\beta f_{\alpha \beta} - \Delta_g \opn{tr}_g(f)\,.
	\end{align}
	Furthermore, if $f_1, f_2 \in \sym$ have mutually disjoint supports and $g_\eps^{\alpha \beta} =g^{\alpha \beta}+ \eps_1 f_1^{\alpha \beta} + \eps_2 f_2^{\alpha \beta}$, then
	\begin{align}
	\partial_{\eps_1} \partial_{\eps_2} K_{g_\eps}(z)|_{(\eps_1,\eps_2)=0} = 0\,.
	\end{align}
\end{lemma}
\begin{proof}
The first claim is proven in Appendix A of \cite{tromba}. The second claim follows directly from the fact that $K_g(z)$ depends only on the metric and its derivatives at the point $z$.
\end{proof}

The inner product \eqref{tangent_ip} leads to the following orthogonal decompositions of the tangent space
\begin{lemma} \label{tangent_decomp}
Let $g \in \met$.
\begin{enumerate}
\item The decomposition $f = f^{\opn{tf}} + \hf \opn{tr}_g(f) g$, where $f^{\opn{tf}} = f - \hf \opn{tr}_g(f)g$, is $L^2$-orthogonal. 
\item There is an $L^2$-orthogonal decomposition $\sym = \sect{}{d} \oplus \sect{}{m}$, where
\begin{align}
\sect{}{d} &= \{f \in T_g\met : f = \mathcal{L}_u g \text{ for some } u \in \Gamma(T^1\Sigma) \}\,, \\
\sect{}{m} &= \{f \in T_g\met: \nabla^\alpha f_{\alpha \beta} = 0 \}\,.
\end{align}
\end{enumerate}
\end{lemma}
\begin{proof}
Lemma 1.2.1 and Theorem 1.4.2 in \cite{tromba}.
\end{proof}
The tensors $f \in \sect{tf}{d}$, thought of as tangent vectors to $\met$, generate perturbations of the metric $g$ that are given by diffeomorphisms, meaning that $g+\eps f = \psi_\eps^* g + \mathcal{O}(\eps^2)$ for some $\psi_\eps \in \dif$. The tangent vectors $f \in \sect{tf}{m}$ generate perturbations of the metric that deform the conformal structure of $\Sigma$, i.e. change the moduli of the surface. The final degree of freedom is the trace part of $f \in \sect{}{}$. These are of the form $f = \varphi g$, where $\varphi \in C^\infty(\Sigma)$, and they generate Weyl transformations of the metric $g+\eps f = e^{\eps \varphi} g + \mathcal{O}(\eps^2)$.

\begin{definition}\label{transv1}
We denote by $T_g^{\opn{tt}}\met$ the subspace of $T_g \met$ consisting of vectors that are orthogonal to the actions of $\dif$ and $C^\infty(\Sigma)$ on $\met$, that is, $f \in T_g^{\opn{tt}}\met$ if $f$ is orthogonal to the tangent spaces $T_g(\dif \cdot g)$ and $T_g(C^\infty(\Sigma) \cdot g)$. Such vectors are called transverse traceless.
\end{definition}

\begin{corollary}\label{transv2}
$\trtr$ can be identified with $\sect{tf}{m}$.
\end{corollary}
\begin{proof}
	Let $f \in \met$. A tangent vector $f \in T_g \met$ is tangent to the orbit $\dif \cdot g$ if 
	\begin{align*}
	f_{\alpha \beta} = (\mathcal{L}_v g)_{\alpha \beta} = \nabla_\alpha v_\beta + \nabla_\beta v_\alpha 
	\end{align*}
	for some smooth vector field $v$ on $\Sigma$, and tangent to the orbit of $C^\infty(\Sigma) \cdot g$ if 
	\begin{align*}
	f_{\alpha \beta} = \varphi g_{\alpha \beta}
	\end{align*}
	for some $\varphi \in C^\infty(\Sigma)$. When we require $f$ to be orthogonal to both of these actions, we get that $f$ has to be traceless with $\nabla^\alpha f_{\alpha \beta} = 0$.
\end{proof}
\begin{remark}\label{rm1}
	For a given $f \in \sect{tf}{}$, the vector field $v$ in the decomposition $f = \mathcal{L}_v g - \hf \opn{tr}(\mathcal{L}_v g)g + f_m$, where $f_m \in \sect{tf}{m}$, is unique up to a vector field $w$ satisfying $\mathcal{L}_w g - \hf \opn{tr}(\mathcal{L}_w g)g = 0$ i.e. a conformal Killing field. Such vector fields do not exist on compact surfaces with $\gen \geq 2$.

	The solution space of the equation $\nabla^\alpha f_{\alpha \beta} = 0$ has complex dimension $3 \gen -3$. It can be identified with the space of holomorphic quadratic differentials on $\Sigma$, i.e. the dual space of the space of harmonic Beltrami differentials on $\Sigma$. This can be identified with the Teichmüller space of $\Sigma$ and with the tangent space of the moduli space of $\Sigma$ at regular (non-orbifold) points.

\end{remark}

\subsection{Conformal Killing operator}\label{conformal_killing}

We define the \emph{conformal Killing operator} $P_g: \Gamma(T_1\Sigma) \to \sect{tf}{}$ by
\begin{align*}
P_g \omega &= 2 S\nabla \omega -  \opn{tr}_g(S\nabla \omega)g = 2\mathcal{L}_{\omega^\sharp}g -  \opn{tr}_g(\mathcal{L}_{\omega^\sharp}g)g\,,
\end{align*}
where $S$ denotes symmetrization and $\omega^\sharp$ is the vector field obtained from $\omega$ by raising the index. $P_g \omega=0$ implies that $\omega^\sharp$ is a conformal Killing field, and thus $P_g$ is injective if $\gen \geq 2$. The adjoint $P_g^*: \sect{tf}{} \to \Gamma(T_1\Sigma)$ is given by
\begin{align*}
(P_g^* f)_\alpha &= -  \nabla^\beta f_{\alpha \beta}\,,
\end{align*}
i.e. $P_g^*f = -\opn{tr}_g(\nabla f)$. By Remark \ref{rm1}, we have $\dim_\C \big( \opn{ran}(P_g)^\perp \big) = \dim_\C \big(\opn{ker}(P_g^*)\big)  = 3\gen-3$, and $\opn{ran}(P_g)= \sect{tf}{d}$. $P_g$ induces an operator $ P_g^\sharp: \Gamma(T^1 \Sigma) \to \Gamma^{\opn{tf}}(ST^2 \Sigma)$ by raising indices. In conformal coordinates
\begin{align*}
(P_g^\sharp u)^{zz} &=  2 \nabla^z u^z \,.
\end{align*}
\begin{remark}
Let $t=t_{z\hdots z}dz^n$ be an $n$-tensor. In conformal coordinates $g= e^\sigma |dz|^2$, and the Christoffel symbols are given by
\begin{align*}
\Gamma^z_{zz} &= \partial_z \sigma\,, \\
\Gamma^z_{z \bar z} &= \Gamma^z_{\bar z z} =  \Gamma^{z}_{\bar z \bar z} = 0\,.
\end{align*}
It follows that
\begin{align*}
\nabla^z t_{z\hdots z} &= g^{z \bar z} \partial_{\bar z} t_{z \hdots z}\,, \\
\nabla_z t_{z\hdots z} &= (\partial_z - n \partial_z \sigma) t_{z\hdots z}\,.
\end{align*}
For $t = t^{z\hdots z} \partial_z^{n}$ we have $\nabla_z t^{z\hdots z} = (\partial_z + n \partial_z \sigma) t^{z\hdots z}$.
\end{remark}
We denote the inverse of $P_g^\sharp$ by $\mathcal{G}: \Gamma^{\opn{tf}}_d(ST^2 \Sigma) \to \Gamma(T^1 \Sigma)$ and we extend it to all of $\Gamma^{\opn{tf}}(ST^2 \Sigma)$ as zero $\mathcal{G}|_{\Gamma^{\opn{tf}}_m(ST^2\Sigma)} := 0$. This implies that for any $f \in \sectt{tf}{}$
\begin{align}
P_g^\sharp \mathcal{G} f &= f-f_m = f_d \,,
\end{align}
where we used the orthogonal decomposition $f=f_{d}+f_m$ induced from the decomposition in (2) of Lemma \ref{tangent_decomp} (where the indices are down). We fix an orthonormal basis $\{\eta_1,\hdots,\eta_{3\gen-3}\}$ of $\sect{tf}{m}$, i.e. 
\begin{align*}
(\eta_i, \eta_j)_g &= \delta_{ij}\,. 
\end{align*}
We denote the integral kernel of $\mathcal{G}$ by $\mathcal{G}^z_{ww}(z,w)$. From the pair of equations
\begin{align*}
f_{d}^{zz} &= f^{zz} - \sum_{k=1}^{3\gen-3} (\eta_k,f)_g \eta_k^{zz}\,, \\
2\nabla^z (\mathcal{G} f)^z(z) &=  2\int_\Sigma f^{ww}(w) \nabla^z \mathcal{G}^z_{ww}(z,w) dv_g(w) = f_d^{zz} \,,
\end{align*}
it follows that
\begin{align}\label{G}
2\nabla^z \mathcal{G}^z_{ww}(z,w) &= \delta_g(z,w)  - \sum_{k=1}^{3\gen-3}  \eta_k^{zz}(z) \eta_{k,ww}(w)\,,
\end{align}
where $\delta_g$ is the Dirac delta distribution on $(\Sigma,g)$. Existence of local solutions to the above equation follows from the theory of the Cauchy transform and the existence of global solutions follows from the fact that the above equation is set up in the way that once one sums together the local solutions the result is orthogonal to the $\eta_k$'s, implying that the sum gives the integral kernel of the inverse of $P_g^\sharp$. Uniqueness follows from the non-existence of conformal Killing fields.

In local conformal coordinates we have
\begin{align*}
\mathcal{G}^z_{ww}(z,w) &= \frac{1}{4\pi} \frac{1}{z-w} + \mathcal{R}(z,w)\,,
\end{align*}
where $\mathcal{R}$ is smooth in $z$ and $w$.

\section{Liouville Conformal Field Theory}\label{section_lcft}

In this Section we summarize basic properties of the Gaussian Free Field (GFF), Gaussian Multiplicative Chaos (GMC) and Liouville Conformal Field Theory (LCFT) on compact Riemann surfaces. None of the results presented in this Section are new. We follow the conventions used in \cite{GRV16}.

\subsection{Gaussian Free Field}

Let $\Sigma$ be an orientable compact surface with genus $\gen \geq 2$ and $g$ a smooth Riemannian metric. The Laplace--Beltrami operator $\Delta_g$ is a self-adjoint operator on $L^2(\Sigma,dv_g)$ with a complete set of orthonormal eigenfunctions $e_{g,n}$,
\begin{align*}
-\Delta_g e_{g,n} &= \lambda_{g,n} e_{g,n}\,, \quad n = 0,1,\hdots\,.
\end{align*}
It holds that $\lambda_{g,n}>0$ for $n>0$ and $\lambda_{g,0}=0$ with $e_{g,0}$ the constant function.

Let $(a_n)_{n \in \N}$ be a sequence of i.i.d. standard Gaussian random variables. The Gaussian Free Field $X_g$ on the surface $(\Sigma,g)$ is defined as the random series
\begin{align}\label{Xseries}
X_g &= \sqrt{2 \pi} \sum_{n=1}^\infty a_n \frac{e_{g,n}}{\sqrt{\lambda_{g,n}}}\,.
\end{align}
The series converges almost surely in the Sobolev space $H^{-s}(\Sigma,dv_g)$ for any $s>0$, see \cite{Dub}. Thus, for any $f \in H^s(\Sigma,dv_g)$, the dual brackets $(X_g,f)_g$ are Gaussian random variables. The covariance of $X_g$ is given by the Green function $G_g$ of $\Delta_g$
\begin{align*}
\E[(X_g,f_1)_g (X_g,f_2)_g] &= \int_{\Sigma^2} G_g(x,y) f_1(x) f_2(y) dv_g(x,y)\,.
\end{align*}
Above $G_g$ is the zero-mean Green function, because we did not include the $n=0$ eigenfunction in the series \eqref{Xseries}, which has the effect of fixing the mean of $X_g$ to be $0$ almost surely.

\subsection{Gaussian Multiplicative Chaos}

The exponential potential in the Liouville action leads us to consider the measure
\begin{align*}
e^{-\mu \int e^{\gamma X_g} dv_g} d\mu_{\opn{GFF}}(X_g)\,,
\end{align*}
where $\mu_{\opn{GFF}}$ is the probability distribution of the GFF. Since $\mu_{\opn{GFF}}$ is supported on a negative order Sobolev space, the definition of $e^{\gamma X_g}$ requires a renormalization The theory of Gaussian Multiplicative Chaos is a general framework for constructing exponentials of $\log$-correlated Gaussian fields, and thus fits into our setting.

Let $\nug{g}{z}$ be the uniform probability measure on the geodesic circle $\partial B_g(z,\delta)$ of radius $\delta > 0 $ at $z \in \Sigma$. We define the circle average regularization of the GFF by
\begin{equation}
X_{g,\delta}(z) := \int X_g(z') d \nug{g}{z}(z')\,.
\end{equation}
The above integral is only formal, since $\nug{g}{z} \notin H^s(\Sigma,dv_g)$ for $s>0$. The precise definition of the above integral can be found in Section 3.2. of \cite{GRV16}.

To define the measure $e^{\gamma X_g} dv_g$, we first introduce the approximations
\begin{align*}
\mathbf{G}^\gamma_{g,\delta}(dx) &:= \delta^{\frac{\gamma^2}{2}} e^{\gamma X_{g,\delta}(x)} dv_g(x)\,.
\end{align*}
For $\gamma \in (0,2)$ the limit
\begin{align*}
\mathbf{G}^\gamma_g &:= \lim_{\delta \to 0} \mathbf{G}^\gamma_{g,\delta}
\end{align*}
exists (and is non-trivial) in the sense of weak convergence of measures almost surely, see \cite{Ber, GRV16}.

\subsection{Liouville expectation}

We define the Liouville expectation of a continuous bounded functional $F: H^{-s}(\Sigma,dv_g) \to \R$, $s>0$, by
\begin{align*}
\langle F \rangle_g &= Z_{\opn{GFF}}(\Sigma,g) \int_\R \E \big[ F(c+X_g) e^{- \frac{Q}{4\pi} \int_\Sigma(c+X_g)K_g dv_g - \mu e^{\gamma c} \mathbf{G}^\gamma_g(\Sigma)} \big] dc\,,
\end{align*}
where $\mu >0$, $\gamma \in (0,2)$, $Q = \frac{2}{\gamma} + \frac{\gamma}{2}$ and $Z_{\opn{GFF}}(\Sigma,g)$ is the so-called partition function of the GFF, defined as 
\begin{align*}
e^{- \hf \zeta'_{\Sigma,g}(0)} v_g(\Sigma)^{- \hf}\,,
\end{align*}
where for $s$ with real part smaller than $-1$
\begin{align*}
\zeta_{\Sigma,g}(s) &= \sum_{n=1}^\infty \lambda_{g,n}^s\,,
\end{align*}
and $\zeta_{\Sigma,g}'(0)$ is defined by analytic continuation, see \cite{Sarnak}.

The next Proposition summarizes the behaviour of the Liouville measure when the background metric is changed by a diffeomorphism or a Weyl transformation. Proof can be found in \cite{GRV16}, Propositions 4.2 and 4.3.
\begin{proposition}\label{F_diffeo_weyl}
Suppose $F: H^{-s}(\Sigma, dv_g) \to \R$ is such that $\langle |F| \rangle_g < \infty$. Then for any $\psi \in \dif$ we have the diffeomorphism covariance
\begin{align*}
\langle F \rangle_{\psi^*g} &= \langle \psi_* F \rangle_g\,,
\end{align*}
where $(\psi_* F)(X) := F(X \circ \psi)$, and for any $\varphi \in C^\infty(\Sigma)$ we have the Weyl anomaly
\begin{align*}
\langle F \rangle_{e^\varphi g} &= e^{c A(\varphi,g)} \langle F(\cdot - \tfrac{Q}{2} \varphi) \rangle_g\,,
\end{align*}
where $A$ is given by \eqref{anomaly}.
\end{proposition}

\subsection{Liouville correlation functions}

In this Section we introduce the correlation functions of the primary fields of LCFT. Formally the primary fields are the exponentials
\begin{align*}
\Vg{g}{\alpha}(x) = e^{\alpha(c+X_g(x))}\,.
\end{align*}
As we learned in the discussion on GMC, to define such exponents, a renormalization is required. Thus, we begin by defining a regularized version of the primary field, and then passing to a limit. The regularized primary field at $x \in \Sigma$ with weight $\alpha \in (-\infty,Q)$ is defined by
\begin{align}\label{vgd}
\Vgd{g}{\alpha}(x) =  \delta^{\frac{\alpha^2}{2}} e^{\alpha (c+X_{g,\delta}(x) )}\,.
\end{align}
Since the regularization of the GFF depends on the metric, we will have to take into account the derivatives $\partial_\eps \Vgd{g_\eps}{\alpha}(x)$ when we differentiate the LCFT correlation functions with respect to the metric in later sections. The regularized correlation function converges to a non-trivial limit, which was shown in \cite{GRV16}.

\begin{proposition}
The correlation functions
\begin{align*}
\langle \prod_j \Vg{g}{\alpha_j}(x_j) \rangle_g &:= \lim_{\delta \to 0} \langle \prod_j \Vgd{g}{\alpha_j}(x_j) \rangle_{g}
\end{align*}
exist and are non-zero if and only if the weights $\alpha_i$ satisfy the Seiberg bounds
\begin{align*}
\sum_j \alpha_j > Q \chi(\Sigma)\,, \\
\alpha_j < Q \quad \forall j\,,
\end{align*}
where $\chi(\Sigma)$ is the Euler characteristic of the surface $\Sigma$. Furthermore, the correlation functions satisfy the diffeomorphism covariance \eqref{intro_diffeo_weyl} and the Weyl anomaly \eqref{intro_weyl}.
\end{proposition}

Primary fields with weight $\alpha = \gamma$ admit a special integral identity, called the KPZ-identity. This was first proven in Lemma 3.3. of \cite{ward} for $\Sigma = \mathbb{S}^2$, but the proof generalizes as such to arbitrary compact Riemann surfaces.
\begin{lemma}\label{kpz} (KPZ-identity)
For all $\delta \geq 0$
\begin{align*}
\int_\Sigma \langle \Vgd{g}{\gamma}(z) \prod_j \Vg{g}{\alpha_j}(x_j) \rangle_{g,\delta} dv_g(z) &= \frac{\sum_j \alpha_j -  \chi(\Sigma)Q}{\mu \gamma} \langle \prod_j \Vgd{g}{\alpha_j}(x_j) \rangle_{g,\delta}\,.
\end{align*}
\end{lemma}

\section{Variation of Green function}\label{section_green}

In this section we consider variations of the Green function of the Laplace--Beltrami operator with respect to the metric $g$. Recall that we sum over repeated Greek indices.

\begin{definition}
	We say that a function $F: \met \to \C$ is differentiable if for all $f \in \sectt{}{}$ and $g_\eps^{\alpha \beta} = g^{\alpha \beta} + \eps f^{\alpha \beta}$ we have
	\begin{equation*}
	\partial_\eps|_0 F(g_\eps) = \int_\Sigma f^{\alpha \beta}(x) F_{\alpha \beta}(g,x) dv_g(x)
	\end{equation*}
for some differentiable functions $x \mapsto F_{\alpha \beta}(g,x): \Sigma \to \C$. If $F$ is differentiable, we denote 
\begin{align*}
\frac{\delta F(g)}{\delta g^{\alpha \beta }(x)} :=F_{\alpha \beta}(g,x)\,.
\end{align*}
\end{definition}

Let $z$ be a local conformal coordinate in which the metric takes the form
\begin{equation*}
g(z) = e^{\sigma(z)}|dz|^2\,.
\end{equation*}
For $f \in \sym$ and $\eps$ small enough it holds that $g_\eps = g + \eps f$ is a Riemannian metric. We will study how the Green function $G_g$ of the Laplace--Beltrami operator $\Delta_g$ is varied under such perturbation of the metric. For a hyperbolic metric $g$ the coordinate $z$ can be chosen to be such that for $z'$ close to $z$ we have
\begin{align*}
G_g(z,z') &= - \ln d_\D(z,z') + M_g(z,z')\,, \\
\Delta_g &= 4e^{-\sigma(z)} \partial_z \partial_{\bar z}\,,
\end{align*}
where $d_\D$  is the hyperbolic distance function on the Poincaré disk, and $M_g$  is a smooth function (see Lemma 2.1. in \cite{GRV16} for proof). The formula for the first variation of $G_g$ is well-known, but for completeness we derive it from scratch for an arbitrary variation of the metric.

\begin{proposition}\label{dotG}
For $x \neq y$, the function $g \mapsto G_g(x,y)$ is differentiable, and for $x \neq z \neq y$ we have
\begin{align}\label{deltag}
\frac{\delta G_g(x,y)}{\delta g^{\alpha \beta}(z)} &= - \frac{1}{4\pi}  \big( \nabla_\alpha G_g(x,z)\nabla_\beta G_g(z,y) + \nabla_\beta G_g(x,z) \nabla_\alpha G_g(z,y) \big) \\
& \quad + \frac{g_{\alpha \beta}(z)}{2 v_g(\Sigma)}  \big(G_g(x,z) + G_g(z,y)\big) \nonumber + \frac{g_{\alpha \beta}(z)}{4\pi}  \nabla_\mu G_g(x,z) \nabla^\mu G_g(z,y)\,.\nonumber
\end{align}
\end{proposition}

\begin{proof}
The Green function solves the Poisson equation
\begin{equation}\label{poisson}
\frac{1}{2\pi} \Delta_g G_g(x,y) = -  \delta_g(x,y) +  \frac{1}{v_g(\Sigma)}\,.
\end{equation}
In coordinates we have $\delta_g(x,y) = \frac{\delta(x-y)}{\sqrt{\det g(x)}}$, where $\delta(x)$ is the (flat) Dirac delta. By plugging in the perturbed metric $g_\eps$, that in conformal coordinate satisfies $g_\eps^{z \bar z} = g^{z \bar z} + \eps f^{z \bar z}$ and $g_\eps^{zz} = \eps f^{zz}$,  and by taking the $\eps$-derivative, we get
\begin{equation}\label{deps} 
\Delta_{g_\eps} \partial_\eps G_{g_\eps}(x,y) = -(\partial_\eps \Delta_{g_\eps}) G_{g_\eps}(x,y) +  \frac{2 \pi \partial_\eps \sqrt{\det g_\eps(y)}}{\sqrt{\det g_\eps(y)}} \delta_{g_\eps}(x,y) - \frac{2\pi \partial_\eps v_{g_\eps}(\Sigma)}{v_{g_\eps}(\Sigma)^2}\,.
\end{equation}
We want to integrate this identity against the Green function $G_{g_\eps}$  over the surface $\Sigma$. We know how the differential operator $\partial_\eps \Delta_{g_\eps}$  looks in conformal coordinates, thus we patch together coordinate expressions using a partition of unity $(U_i, \rho_i)$. Denote $m_g(f) = \frac{1}{v_g(\Sigma)} \int f dv_g$. By integrating in \eqref{deps} both sides against $G_{g_\eps}$ and by using \eqref{poisson} and the zero-mean property of $G_g$
\begin{align*}
\int G_g(x,z) dv_g(z) &= 0\,,
\end{align*}
we get
\begin{align}\label{g_var_eps}
&- 2 \pi \partial_\eps G_{g_\eps}(x,y) + 2 \pi m_{g_\eps}(\partial_\eps G_{g_\eps}(\cdot,y)) \\
&= -\lim_{\delta \to 0} \int_{d_g(z,y)>\delta} G_{g_\eps}(x,z) (\partial_\eps \Delta_{g_\eps}) G_{g_\eps}(z,y) dv_{g_\eps}(z)   + 2\pi \lim_{\delta \to 0} \int_{d_g(z,y)>\delta} G_{g_\eps}(x,z) \frac{\partial_\eps \sqrt{\det g_\eps(z)}}{ \sqrt{\det g_\eps(z)}} \delta_{g_\eps}(z,y) dv_{g_\eps}(z) \nonumber \\
&= - \lim_{\delta \to 0} \sum_i \int_{U_i} \mathbf{1}_{d_g(z,y)>\delta} \rho_i(z) G_{g_\eps}(x,z) \partial_\eps \big( \tfrac{1}{\sqrt{\det g_{\eps}(z)}} \partial_\alpha(\sqrt{\det g_\eps(z)} g_\eps^{\alpha \beta}(z) \partial_\beta) \big) G_{g_\eps}(z,y) dv_{g_\eps}(z) \nonumber \\
& \quad + 2 \pi G_{g_\eps}(x,y) \frac{\partial_\eps \sqrt{\det g_\eps (y)}}{\sqrt{\det g_\eps (y)}}\,. \nonumber 
\end{align} 
Note that by using the zero-mean property of $G_g$ again we get
\begin{align*}
m_{g_\eps}(\partial_\eps G_{g_\eps}(\cdot,y)) &= \partial_\eps m_{g_\eps}(G_{g_\eps}(\cdot,y)) - \frac{\partial_\eps v_{g_\eps}(\Sigma)}{v_{g_\eps}(\Sigma)} m_{g_\eps}(G_{g_\eps}(\cdot,y)) -\frac{1}{v_{g_\eps}(\Sigma)} \int G_{g_{\eps}}(z,y) \partial_\eps \sqrt{\det g_\eps (z)} d^2z \\
&= -\frac{1}{v_{g_\eps}(\Sigma)} \int G_{g_{\eps}}(z,y) \partial_\eps \sqrt{\det g_\eps (z)} d^2z\,.
\end{align*}
Next we set $\eps = 0$. Consider first the case $\opn{tr}_g(f) = 0$, i.e. $f^{z \bar z} = 0$. Then we have
\begin{align*}
\sqrt{\det g_\eps(z)} = \sqrt{\det g(z)} + \mathcal{O}(\eps^2)\,.
\end{align*}
We get 
\begin{align*}
\partial_\eps|_0 G_{g_\eps}(x,y) &=   \frac{1}{2\pi} \lim_{\delta \to 0} \sum_i \int_{U_i} \mathbf{1}_{|z-y|>\delta} \rho_i(z) G_g(x,z) \partial_z \big(e^{\sigma(z)} f^{zz}(z) \partial_z \big) G_g(z,y)  d^2z\,.
\end{align*}
Next we integrate by parts the $\partial_z$. The boundary term $\partial_z \mathbf{1}_{|z-y|>\delta}$ vanishes as $\delta \to 0$, since $\int_{|z|=\delta} \frac{h(z)}{z} d \bar z = o_\delta(1)$ for smooth $h$. Note that here we need to assume $x \neq y$ to ensure that $G_g(x,z)$ is smooth when $z$ is close to $y$. After passing to the $\delta \to 0$ limit, we get
\begin{align*}
\partial_\eps|_0 G_{g_\eps}(x,y) &=- \frac{1}{2\pi}  \sum_i \int_{U_i} \rho_i(z) f^{zz}(z) \partial_z G_g(x,z) \partial_z G_g(z,y) e^{\sigma(z)} d^2z  \\
& \quad -\frac{1}{2\pi} \sum_i \int_{U_i} \partial_z \rho_i(z) f^{zz} (z) G_g(x,z) \partial_z G_g(z,y) e^{\sigma(z)} d^2z  \\
& \quad +\frac{1}{2\pi} \sum_i \int_{\partial U_i} \rho_i(z)  f^{zz}(z) G_g(x,z) \partial_z G_g(z,y) e^{\sigma(z)}  d \bar z \,.
\end{align*}
The last two terms vanish since $(U_i, \rho_i)$ is a partition of unity and the surface $\Sigma$ has no boundary. Now it follows that $\partial_\eps|_0 G_{g_\eps}(x,y)$ is a smooth function outside the diagonal.

Next consider the trace part of $f$, i.e. let $f^{\alpha \beta} = \hf \opn{tr}_g(f) g^{\alpha \beta}$. Now $f^{zz} = f^{\bar z \bar z} = 0$, and using $\sqrt{\det g_\eps} = \sqrt{\det g} - 2 \eps g_{z \bar z}^2 f^{z \bar z} + \mathcal{O}(\eps^2)$ we get 
\begin{align*}
\partial_{\eps}|_0 G_{g_\eps}(x,y) &= \frac{1}{v_g(\Sigma)} \int G_g(z,y) g_{z \bar z}(z) f^{z \bar z} (z) dv_g(z) \\
& \quad - \frac{1}{2\pi} \sum_i \int_{U_i} \rho_i(z) f^{z \bar z}(z) \big(\partial_z G_g(x,z) \partial_{\bar z} G_g(z,y) + \partial_{\bar z} G_g(x,z) \partial_z G_g(z,y) \big) dv_g(z) \\
& \quad + \frac{1}{2\pi} \sum_i \int_{U_i} \rho_i(z)2 f_{z \bar z}(z) \partial_\alpha G_g(x,z)   g^{\alpha \beta} \partial_\beta G_g(z,y) d^2z   \\
& \quad + \frac{1}{2\pi} \sum_i \int_{U_i} \rho_i(z) 2 f_{z \bar z}(z) G_g(x,z) \frac{1}{\det g} \partial_\alpha \big(\sqrt{\det g} g^{\alpha \beta} \partial_\beta G_g(z,y)\big) dv_g(z) \\
& \quad + 2 G_g(x,y) \frac{f_{y \bar y}(y)}{\sqrt{\det (y)}}
\end{align*}
The second to last term above can be simplified
\begin{align*}
& \frac{1}{2\pi} \sum_i \int_{U_i} \rho_i(z) 2 f_{z \bar z}(z) G_g(x,z) \frac{1}{\det g(z)} \partial_\alpha \big(\sqrt{\det g} g^{\alpha \beta} \partial_\beta G_g(z,y)\big) dv_g(z) \\
&=  \frac{1}{2\pi} \sum_i \int_{U_i} \rho_i(z) 2 f_{z \bar z} (z) G_g(x,z) \frac{1}{\sqrt{\det g(z)}} \Delta_g G_g(z,y) dv_g(z) \\
&=- 2 G(x,y) \frac{f_{y \bar y}(y)}{\sqrt{\det y}}  + \frac{1}{v_g(\Sigma)}\int g_{z \bar z}(z) f^{z \bar z}(z) G_g(x,z) dv_g(z)\,.
\end{align*}
Now the result follows.
\end{proof}

The above result becomes clearer when viewed in terms of the resolvent $R_g(\lambda) = (\Delta_g-\lambda)^{-1}$ and the second resolvent formula
\begin{align}\label{resolvent_formula}
R_{g_\eps}(\lambda) - R_g(\lambda) &= -R_{g_\eps}(\lambda)(\Delta_{g_\eps}-\Delta_g) R_g(\lambda)\,,
\end{align}
which holds for any $\lambda$ belonging to resolvent sets of both $\Delta_{g_\eps}$ and $\Delta_g$. For small $\eps$ this holds for any $0<|\lambda|<\delta$ for $\delta$ sufficiently small. The resolvent formula then implies (we denote $\dot \Delta_g = \partial_\eps|_0 \Delta_{g_\eps}$ and $\dot R_g = \partial_\eps|_0 R_{g_\eps}$)
\begin{align*}
\dot R_g(\lambda) &= -R_g(\lambda) \dot \Delta_g R_g(\lambda)\,,
\end{align*}
which is a well-defined operator on $L^2(\Sigma,g)$. We also see that $\dot R_g(\lambda)$ is holomorphic in $\lambda$ in the punctured ball $B(0,\delta) \setminus\{0\} \subset \C$ and thus an application of the Cauchy integral formula implies
\begin{align*}
\dot R_g &= R_g \dot \Delta_g R_g\,,
\end{align*}
where $R_g: L^2(\Sigma,g) \to L^2(\Sigma,g)$ is the operator that has the zero-mean Green function $G_g$ as its integral kernel. It satisfies $\Delta_g R_g = \opn{Id}-\Pi_0$, where $\Pi_0$ is the projection onto constants.

Now, we can view Proposition \ref{dotG} as the computation of the integral kernel of $\dot R_g$. We saw that for traceless variations we get exactly the same operator as from the resolvent formula, but for variations with trace, one has to be more careful and take into account the variation of the volume form of the metric. For now, we will only look at traceless variations, and variations with trace will be included by applying transformation properties of $G_g$ under Weyl transformations (see Corollary \ref{corollary}).

%In other words (*)
%\begin{align*}
%g_\eps = g +\sum_{k=1}^n 2 \opn{Re}(\eps_k f_k) (dx^2-dy^2)-4 \opn{Im}(\eps_k f_k) dxdy\,.
%\end{align*}
%%%This implies that f is real, trace-free and symmetric
%%%To add a trace, put \sum_{k=1}^n 2\opn{Re}(\eps) f_k^{z \bar z} (dx^2 +dy^2)

\begin{proposition}\label{ddotG}
Let $f \in \sectt{tf}{}$ and $g_\eps^{zz} = \eps f^{zz}$, $g_\eps^{z \bar z} = g^{z \bar z}$. Then, for any $\lambda$ belonging to the resolvent set of $\Delta_g$, and $h \in C^\infty(\Sigma)$ supported in $\Sigma \setminus \{x\}$, we have
\begin{align*}
\partial_\eps^n|_0 (R_{g_\eps}(\lambda) h)(x) &= (-1)^n n! \big( R_g(\lambda) (\dot \Delta_g  R_g(\lambda))^n h \big)(x)\,,
\end{align*}
where $\dot R_g(\lambda) = \partial_\eps|_0 R_{g_\eps}(\lambda)$ and $\dot \Delta_g = \partial_{\eps}|_0 \Delta_{g_\eps}$. The operator $\partial_\eps^n|_0 R_{g_\eps}$ has an integral kernel $\dot G_g^{(n)}(\lambda;\cdot,\cdot) \in C^\infty(\Sigma^2 \setminus \opn{diag})$,
\begin{align*}
\partial_\eps^n|_0 (R_{g_\eps}(\lambda) h)(x) &= \int \dot G_g^{(n)}(\lambda;x,y) h(y) dv_g(y)\,,
\end{align*}
for all $\lambda \in B(0,\delta)$ with $\delta$ small enough, where $R_g(0) = R_g$.
\end{proposition}
\begin{proof}
We proceed by induction. The $n=1$ case follows from Proposition \ref{dotG}.

Assume that the claim holds for some $n \in \N$. For traceless perturbations $\partial_\eps^k|_0 \sqrt{\det g_\eps} = 0$ for all $k \in \N$ (traceless perturbations only contribute to the mixed derivatives $\partial_\eps^k \partial_{\bar \eps}^k|_0 \sqrt{\det g_\eps}$). Thus, we do not have to carry the $\eps$-derivatives of the volume form in our computations. This means that we can use the resolvent formula \eqref{resolvent_formula} to compute the variation, and there are no additional terms coming from the volume form. By using the resolvent formula \eqref{resolvent_formula} and $\partial_\eps^2|_0 \Delta_{g_\eps} = 0$, we get
\begin{align*}
\partial_\eps^{n+1}|_0 \big(R_{g_\eps}(\lambda) h \big)(x) &= \partial_\eps^{n+1}|_0 (- R_{g_\eps}(\lambda)(\Delta_{g_\eps}-\Delta_g) R_g(\lambda))  \\
&= -(n+1)\big( \partial_\eps^n|_0 R_{g_\eps}(\lambda) \dot \Delta_g R_g(\lambda) h \big) (x)\,.
\end{align*}
Now the first claim follows.

Next we study the integral kernel. We denote the integral kernel of $R_g(\lambda)$ by $G_g(\lambda;\cdot,\cdot)$. The behaviour close to the diagonal is given by $G(\lambda;x,y) \sim - \ln d_g(x,y)$, see Section 2.4 of \cite{GRV16}. From the induction hypothesis we get
\begin{align}\label{qwert}
\big( \partial_\eps^n|_0 R_{g_\eps}(\lambda) \dot \Delta_g R_g(\lambda) h \big)(x) &= \int \dot G_g(\lambda;x,z) (\dot \Delta_g R_g(\lambda) h)(z) dv_g(z) \nonumber \\
&= \lim_{\delta \to 0} \int_{d_g(z,y) > \delta}  \dot G_g^{(n)}(\lambda; x,z) \dot \Delta_g(z) G_g(\lambda;z,y) h(y) dv_g(z,y)\,. 
\end{align}
Above we used 
\begin{align*}
\partial_z^2  \int \log|z-y| h(y) d^2y &= \lim_{\delta \to 0} \int_{|z-y|>\delta}  \frac{-1}{2(z-y)^2} h(y) d^2y\,,
\end{align*}
which holds for any smooth and compactly supported $h$. The $\delta \to 0$ limit of \eqref{qwert} exists, which we demonstrate next. Let $z \mapsto \chi(x;z)$ be a smooth bump function supported in a small neighbourhood of $x$. Then 
\begin{align*}
&\int_{d_g(z,y) > \delta}  \dot G_g^{(n)}(\lambda; x,z) \dot \Delta_g G_g(\lambda;z,y) h(y) dv_g(z,y) \\
&= \int_{d_g(z,y) > \delta} \chi(x,z)  \dot G_g^{(n)}(\lambda; x,z) \dot \Delta_g G_g(\lambda;z,y) h(y) dv_g(z,y) \\
& \quad + \int_{d_g(z,y)>\delta} \dot \Delta_g \big( (1-\chi(x,z)) \dot G_g^{(n)}(\lambda;x,z) \big) G_g(\lambda;z,y) h(y) dv_g(z,y) \\
&\quad + r_\delta(x)
\end{align*} 
where $r_\delta$ contains boundary terms from integrating by parts the $\dot \Delta_g$
\begin{align*}
r_\delta(x) &= \int \Big( \int_{|z-y|=\delta} (1-\chi(z,x))\dot G_g^{(n)}(\lambda;x,z) f^{zz}(z) \partial_z G_g(\lambda;z,y) \sqrt{\det g(z)}d \bar z \Big) h(y) dv_g(y) \\
& \quad + \int \Big( \int_{|z-y|=\delta} \partial_z \big( (1-\chi(z,x)) \dot G_g^{(n)}(\lambda;x,z) \big) f^{zz}(z) G_g(\lambda;z,y) \sqrt{\det g(z)} d \bar z \Big) h(y) dv_g(y)\,.
\end{align*}
The first term vanishes as $\delta \to 0$ by first Taylor expanding the smooth terms around $z=y$ and by observing that $\int_{|z-y|=\delta} \frac{1}{z-y} d \bar z = 0$. The second term works similarly but with $\int_{|z-y|=\delta} \log|z-y| d \bar z = 0$.

Now we have shown that
\begin{align*}
\big( \partial_\eps^n|_0 R_{g_\eps}(\lambda) \dot \Delta_g R_g(\lambda) h \big)(x) &= \int \dot G_g^{(n)}(\lambda;x,z) F_1(\lambda;x,y,z) h(y) dv_g(z,y) \\
& \quad + \int G_g(\lambda;z,y) F_2(\lambda;x,y,z) h(y) dv_g(z,y)
\end{align*}
for some smooth functions $F_1(\lambda; \cdot)$ and $F_2(\lambda,\cdot)$ and for any $h$ supported away from $x$. Because both $\dot G_g^{(n)}$ and $G_g$ are smooth outside the diagonal, and thus we can define $\dot G_g^{(n+1)} = \dot G_g^{(n)} F_1 + G_g F_2$.

The claim for $\lambda=0$ follows by taking a contour integral around the origin and using the Cauchy integral formula.
\end{proof}

Next we control the circle averages of the integral kernels $\dot G_g^{(n)}(x,y)$ when $x$ is close to $y$. This can be done in local coordinates, so the perturbed metric can be assumed to take the form $g_\eps = \psi_\eps^*(e^{\varphi_\eps}|dz|^2)$, where $\psi_\eps$ is some locally defined diffeomorphism and $\varphi_\eps$ is some locally defined smooth function. Recall that we denote by $\nug{g}{x}$ the uniform probability measure on the geodesic circle $\partial B_g(x,\delta)$ at $x \in \Sigma$ with radius $\delta > 0$.

\begin{proposition}\label{circleaverage}
Let $f \in \sectt{tf}{}$ and set $g_\eps^{zz} = \eps f^{zz}$, $g_\eps^{z \bar z} = g^{z \bar z}$. Denote\footnote{It will turn out that $\dot G_g^{(n)}(x,y) = \dot G_g^{(n)}(\lambda;x,y)$ for $x \neq y$, so there is no risk of confusion in the notation. Note that $\dot G_g^{(n)}(\lambda;x,x)$ is not defined.}
\begin{align*}
\dot G_g^{(n)}(x,y) &:=  \lim_{\delta \to 0} \partial_\eps^n|_0 \int G_{g_\eps}(x',y') d \nug{g_\eps}{x}(x') d \nug{g_\eps}{y}(y')\,.
\end{align*}
Then $\dot G_g^{(n)} \in L^\infty(\Sigma^2)$ and $\partial_x^a \partial_y^b \dot G_g^{(n)} \in C^\infty(\Sigma^2\setminus \opn{diag})$ for all $(a,b) \in \N \times \N$. The restriction to the diagonal $x \mapsto \dot G_g^{(n)}(x,x)$ belongs to $C^\infty(\Sigma)$.
\end{proposition}
\begin{proof}
For a smooth function $F:\Sigma \times \Sigma \to \C$ we have
\begin{align*}
\partial_\eps^n|_0 \int F(x',y') d \nug{g_\eps}{x}(x') d\nug{g_\eps}{y}(y') &= o_\delta(1)\,.
\end{align*}
Since the derivatives $\partial_\eps^n|_0 G_{g_\eps}(x,y)$ are smooth outside the diagonal by Proposition \ref{ddotG}, it suffices to consider $(x,y)$ taken from a small neighbourhood of the diagonal of $\Sigma \times \Sigma$. Thus, we can assume that $x$ and $y$ belong to some small open set $U \subset \Sigma$ and in $U$ the metric is $g_\eps = \psi_\eps^*(e^{\varphi_\eps} |dz|^2)$, where $\psi_\eps$ and $U$ are chosen in such a way that $\psi_\eps: U \to V \subset \C$, $\psi_\eps(x) = p \in V$ for all $\eps \geq 0$. Thus, in $U$, the Green function can be written as
\begin{align}\label{GdM}
G_{g_\eps}(x,y) &= -\ln d_{\psi_\eps^*(e^{\varphi_\eps} g_0)}(x,y) + M_{g_\eps}(x,y)\,,
\end{align}
where $g_0$ is the Euclidean metric on $U$. We can simplify the first term on the right-hand side
\begin{align*}
\ln d_{\psi_\eps^* (e^{\varphi_\eps} g_0)}(x,y) &= \hf \varphi_\eps( \psi_\eps(x)) + \ln |\psi_\eps(x)-\psi_\eps(y)|  + r(\eps,x,y)\,,
\end{align*}
where $r$ is smooth and $\partial_\eps^k|_0 r = \mathcal{O}(d_g(x,y))$ as $y \to x$ for all $k \geq 0$, see Appendix \ref{appendix}.

We first focus on the logarithm of the distance function. It holds that $\psi^* \nug{g}{\psi(x)} = \nug{\psi^*g}{x}$. For $x \neq y$ we have 
\begin{align*}
\int \ln d_{g_\eps}(x',y') d \nug{g_\eps}{x}(x') d \nug{g_\eps}{y}(y') &=  \int\big( \hf \varphi_\eps(x')+ \ln |x'-y'| \big)  d \nug{e^{\varphi_\eps} g_0}{p}(x') d \nug{e^{\varphi_\eps} g_0}{\psi_\eps(y)}(y') + r_\delta(\eps,x,y)\,.
\end{align*}
Denote $\delta(\eps,q) = \delta e^{-\frac{\varphi_\eps(q)}{2}}$. For small $\delta$, the geodesic circles $\partial B_{e^\varphi g}(x,\delta)$ and $\partial B_g(x, e^{- \varphi(x)/2}\delta)$ are almost the same. This allows us to write
\begin{align*}
\int \big( \hf \varphi_\eps(x')+ \ln |x'-y'| \big)  d \nug{e^{\varphi_\eps} g_0}{p}(x') d \nug{e^{\varphi_\eps} g_0}{\psi_\eps(y)}(y') &= \hf \varphi_\eps(p)+ \int \ln |x'-y'| d \nu^{\delta(\eps,p)}_{g_0,p}(x') d \nu^{\delta(\eps, \psi_\eps(y))}_{g_0,\psi_\eps(y)}  \\
& \quad + \tilde r_\delta(\eps,x,y)\,.
\end{align*}
where  $\partial_\eps^k|_0 \tilde r_\delta = o_\delta(1)$ for all $k \in \N$, which follows essentially from the computation done in Appendix \ref{appendix}. Recall that $g_0$ is the Euclidean metric. The circle average becomes
\begin{align*}
\int \ln |x'-y'| d \nu^{\delta(\eps,p)}_{g_0,p}(x') d \nu^{\delta(\eps, \psi_\eps(y))}_{g_0,\psi_\eps(y)}&=  \int \ln |p+ \delta(\eps,p)e^{i \theta}-\psi_\eps(y) - \delta(\eps,\psi_\eps(y)) e^{i \theta'}| \frac{d \theta d \theta'}{(2\pi)^2} \,.
\end{align*}%just use that \ln|x| is harmonic (*)
Since $\psi_\eps(y) \to q \neq p$ as $\eps \to 0$ (because we assumed $x \neq y$), we can Taylor expand 
\begin{align*}
 \ln \big( p - \psi_\eps(y) +  ( e^{- \frac{\varphi_\eps(p)}{2}+ i \theta} - e^{- \frac{\varphi_\eps(\psi_\eps(y))}{2} + i \theta'} )\delta\big) &= \ln (p - \psi_\eps(y))   - \sum_{k=1}^\infty (-\delta)^k \frac{(e^{- \frac{\varphi_\eps(p)}{2}+ i \theta} - e^{- \frac{\varphi_\eps(\psi_\eps(y))}{2} + i \theta'})^k}{k (p-\psi_\eps(y))^k}
\end{align*}
and it follows that
\begin{align*}
\partial_\eps^n|_0 \int \ln |p+ \delta(\eps,p)e^{i \theta}-\psi_\eps(y) - \delta(\eps,\psi_\eps(y)) e^{i \theta'}| \frac{d \theta d \theta'}{(2\pi)^2} &= \partial_\eps^n|_0 \ln |p-\psi_\eps(y)| \\
&= \partial_\eps^n|_0 \ln |\psi_\eps(x)-\psi_\eps(y)|\,,
\end{align*}
which is bounded together with its derivatives on $\Sigma^2 \setminus \opn{diag}$. In conclusion,
\begin{align}\label{lndgca}
\int \ln d_{g_\eps}(x',y') d \nug{g_\eps}{x}(x') d \nug{g_\eps}{y}(y') &= \hf \varphi_\eps(\psi_\eps(x)) + \ln |\psi_\eps(x)-\psi_\eps(y)| + \tilde r_\delta(\eps,x,y)\,,
\end{align}
where $\tilde r_\delta$ is smooth and $\partial_\eps^k|_0 \tilde r_\delta = o_\delta(1)$.

In the case $x=y$ we of course have $\psi_\eps(x) = \psi_\eps(y)$, and we can directly conclude that
\begin{align}\label{lndgca2}
\partial_\eps^n|_0 \int \ln d_{g_\eps}(x',y') d \nug{g_\eps}{x}(x') d \nug{g_\eps}{x}(y') &= \partial_\eps^n|_0 \big( \hf \varphi_\eps(p) + \int \ln |\delta(\eps,p)(e^{i \theta}-e^{i \theta'})| \frac{d \theta d \theta'}{(2\pi)^2}\big) + o_\delta(1) \\
&= \partial_\eps^n|_0 \big( \hf \varphi_\eps(p) - \hf \varphi_\eps(p) \big) + o_\delta(1) \nonumber \\
&= o_\delta(1)\,, \nonumber
\end{align}
where the $o_\delta(1)$-terms are such that also their derivatives in $x$ vanish as $\delta \to 0$ (by Appendix \ref{appendix}). Note that from equations \eqref{lndgca} and \eqref{lndgca2} it can be inferred that the function
\begin{align*}
(x,y) \mapsto \lim_{\delta \to 0} \partial_\eps^n|_0 \int \ln d_{g_\eps}(x',y') d \nug{g_\eps}{x}(x') d \nug{g_\eps}{y}(y')
\end{align*}
is discontinuous at the diagonal whenever $\psi_\eps$ is not a conformal map, but we do not go into the details of this.

In the case that $g_\eps$ is hyperbolic for all $\eps \geq 0$, this result is expected since in Lemma 3.2. of \cite{GRV16} it was computed that 
\begin{align*}
\int \ln d_g(x',y') d\nug{g}{x}(x') d \nug{g}{x}(y') = \ln |2 \tanh(\tfrac{\delta}{2})| + o_\delta(1)
\end{align*}
for any hyperbolic metric $g$, so the relevant contribution is independent of the (hyperbolic) metric. For non-hyperbolic perturbations we have to carry the $\varphi_\eps$ in the computation, but in the end the $\partial_\eps^n|_0$-derivative is still $o_\delta(1)$, as we just showed.

Now we have shown that the function
\begin{align*}
(x,y) \mapsto \lim_{\delta \to 0} \partial_\eps^n|_0 \int \ln d_{g_\eps}(x',y') d \nug{g_\eps}{x}(x') d \nug{g_\eps}{y}(y')
\end{align*}
together with its derivatives belong to $L^\infty(\Sigma^2)$ and the diagonal value vanishes.

Finally, we show that the contribution to the circle average from $M_g$ in \eqref{GdM} is smooth. To this end, we use the setting of Proposition 2.2. in \cite{GRV16}, where the resolvent of the hyperbolic Laplacian $\Delta_\D$ was used as a parametrix for the resolvent of $\Delta_g$. Following Section 2.4. of \cite{GRV16}, we write the integral kernel of the resolvent of $\Delta_\D$ as
\begin{align*}
G_\D(\lambda;x,y) &= F_\lambda(d_\D(x,y))\,,
\end{align*}
where $\lambda \mapsto F_\lambda(r)$ is holomorphic for $r > 0$ and 
\begin{align*}
F_\lambda(r) = \mathcal{O}(\log r)
\end{align*}
as $r \to 0$ for all $\lambda \in B(0,1/4) \subset \C$ and
\begin{align*}
F_0(r) &= - \log r + m(r^2)
\end{align*}
for some smooth function $m$.

Let $\chi$ be a smooth bump function around $x$ and let $\tilde \chi$ be slightly larger bump function such that $\tilde \chi \chi = \chi$. We assume the support of $\chi$ to be large enough to also contain the point $y$ (which we have assumed to be close to $x$ in the beginning). Let $L_g(\lambda)$ be the integral operator with integral kernel $F_\lambda(d_g(x,y))$. Then, in Proposition 2.2. of \cite{GRV16} it was shown that
\begin{align}\label{parametrix}
\chi R_g(\lambda) \chi &= \chi L_g(\lambda) \chi - \chi R_g(\lambda) [\Delta_g, \tilde \chi] L_g(\lambda) \chi\,.
\end{align} 
Here it is important to observe that for any Sobolev function $f \in H^1(\Sigma,dv_g)$, the function $[\Delta_g,\tilde \chi]f$ is supported outside the support of $\chi$. Thus, on the right-hand side above the latter term is a smoothing operator that only depends on the integral kernels of $R_g(\lambda)$ and $L_g(\lambda)$ outside the diagonal. Thus, when we restrict to functions supported in $\opn{supp}(\chi)$, the integral kernel $F_\lambda(d_g(x,y))$ approximates $G_g(x,y)$ up to some smooth function that also depends smoothly on the metric (Proposition \ref{ddotG}).

Taking a contour integral over a small circle around $\lambda=0$ of \eqref{parametrix} implies that near the diagonal $G_g(x,y)$ is given in terms of the diagonal values of $F_0(d_g(x,y))$ and the off-diagonal values of $G_g$ and $G_\D$. In the proof of Proposition 2.2. in \cite{GRV16} it was shown that the integral kernel of $\chi L_g(0) \chi$ is given by
\begin{align*}
\chi(x) \chi(y) F_0(d_g(x,y)) = -  \chi(x) \chi(y) \log d_g(x,y) + \chi(x) \chi(y) m(d_g(x,y)^2)
\end{align*}
for some smooth function $m$. Now, for $x,y \in \opn{supp}(\chi)$, the function $M_{g_\eps}$ appearing in Equation \eqref{GdM} satisfies
\begin{align*}
\partial_\eps^n|_0 \int M_{g_\eps}(x',y') d \nug{g_\eps}{x}(x') d \nug{g_\eps}{y}(y') &= \partial_\eps^n|_0 \int m(d_{g_\eps}(x,y)^2) d \nug{g_\eps}{x}(x') d \nug{g_\eps}{y}(y') \\
& \quad + \partial_\eps^n|_0 \int G_{g_\eps}(x',z) [\Delta_g,\tilde \chi](z) G_\D(z,y') d \nug{g_\eps}{x}(x') d \nug{g_\eps}{y}(y')\,.
\end{align*}
We do not have to care about the circle-average since the integrands are smooth. Since $m$ is smooth and $\eps \mapsto d_{g_\eps}(x,y)^2$ is smooth when $x$ and $y$ are close to each other, it follows that the above expression defines a smooth function in $(x,y)$. By combining this with Equations \eqref{lndgca} and \eqref{lndgca2}, we are done.
\end{proof}
The result of Proposition \ref{circleaverage} generalizes to arbitrary perturbations of the metric.

\begin{corollary}\label{corollary}
The statement of Proposition \ref{circleaverage} holds for all $f \in \sectt{}{}$.
\end{corollary}
\begin{proof}
If $\opn{tr}_g(f) \neq 0$, then $g_\eps = e^{\varphi_\eps} \tilde g_\eps$, where $\varphi_\eps = \tfrac{\eps}{2} \opn{tr}_g(f) + \mathcal{O}(\eps^2)$, $\tilde g_\eps^{zz} = \eps (f^{\opn{tf}})^{zz}$, $\tilde g_\eps^{z \bar z} = g^{z \bar z}$. The Green function has the following transformation property under Weyl transformations $g \mapsto e^\varphi g$
\begin{align}\label{green_weyl}
G_{e^\varphi g}(x,y) &= G_g(x,y) - \frac{1}{v_{e^\varphi g}(\Sigma)}\int_\Sigma G_g(x,z) e^{\varphi(z)} dv_{ g}(z) - \frac{1}{v_{e^\varphi g}(\Sigma)} \int_\Sigma G_g(z,y) e^{\varphi(z)} dv_g(z) \nonumber \\
& \quad + \frac{1}{v_{e^\varphi g}(\Sigma)^2}\int_{\Sigma^2} G_g(z,z') e^{\varphi(z) + \varphi(z')}dv_{g}(z,z')\,.
\end{align}
We plug the above formula in the case of $G_{e^{\varphi_\eps} \tilde g_\eps}$ into the circle average in the statement of Proposition \ref{circleaverage}. The derivatives
\begin{align*}
&\partial_\eps^n|_0 \int G_{\tilde g_\eps}(x',y') d \nug{\tilde g_\eps}{x}(x') d \nug{\tilde g_\eps}{y}(y') \,, \\
&\partial_\eps^n|_0 \int G_{\tilde g_\eps}(x',z) d\nug{\tilde g_\eps}{x}(x') e^{\varphi_\eps(z)} dv_{\tilde g_\eps}(z)
\end{align*}
can be treated similarly as we did in Proposition \ref{circleaverage}, because $\varphi_\eps$ depends smoothly on $\eps$. The $\eps$-derivatives of the last term in \eqref{green_weyl} are well-defined and independent of $x$ and $y$. Thus, the claim follows.
\end{proof}

\section{Variation of correlation functions}\label{section_correlation}

In this section we show that the derivatives of the LCFT correlation functions with respect to the metric exist, i.e. we study the map
\begin{align*}
f \mapsto \partial_\eps^n|_0 \langle \prod_j \Vg{g_\eps}{\alpha_j}(x_j) \rangle_{g_\eps}\,,
\end{align*}
where $g_\eps^{\alpha \beta} = g^{\alpha \beta} + \eps f^{\alpha \beta}$ for $f \in \sectt{}{}$. In this section we show that the derivatives exist. As an application of this, in the next section, we consider variations of the form $g_\eps^{\alpha \beta} = g^{\alpha \beta} + \sum_{k=1}^n \eps_k f_k^{\alpha \beta}$, where $f_k \in \sectt{}{}$ for each $k$, and the supports of the $f_k$'s are mutually disjoint. We show that the derivatives can be expressed in the form
\begin{align}\label{Tmunu}
(4\pi)^n  \prod_{k=1}^n \partial_{\eps_k}|_0 \langle \prod_j \Vg{g_\eps}{\alpha_j}(x_j) \rangle_{g_\eps} &= \int_{\Sigma^n} \prod_{k=1}^n f_k^{\mu_k \nu_k}(z_k) \langle \prod_{k=1}^n T_{\mu_k \nu_k}(z_k) \prod_j \Vg{g}{\alpha_j}(x_j) \rangle_g dv_g(\bz)
\end{align}
for some functions $\langle \prod_{k=1}^n T_{\mu_k \nu_k}(z_k) \prod_j \Vg{g}{\alpha_j}(x_j) \rangle_g$ that are smooth in the region of non-coinciding points. The object $T_{\mu \nu}$ is called the \emph{stress-energy tensor} (SE-tensor), even though precisely speaking for us it is only defined as a notation in the sense of \eqref{Tmunu}.

We write the regularized correlation function as
\begin{align}
\langle \prod_j \Vgd{g_\eps}{\alpha_j}(x_j) \rangle_{g_\eps,\delta} &= Z_{\opn{GFF}}(\Sigma,g_\eps) \int_\R \E_{g_\eps} [F_{g_\eps,\delta}(X+c)] \, dc\,,
\end{align}
where
\begin{equation}
F_{g,\delta} (X) = \prod_j \Vgd{g}{\alpha_j}(x_j) e^{- \frac{Q}{4\pi} \int K_g X_{g,\delta} dv_g - \mu \int V_{g,\delta}^\gamma dv_g }\,.
\end{equation}
Then we apply the change of covariance formula (see Proposition 9.2.1. in \cite{GJ})
\begin{align}\label{variation_of_gaussian}
\partial_\eps \E_{g_\eps} [	 F_{g,\delta}(X)] &= \hf \int_{\Sigma^2} \partial_\eps \big( G_{g_\eps}(z,y) \big) \E_{g_\eps} \big[ \tfrac{\delta^2 F_{g,\delta}(X)}{\delta X(z) \delta X(y)} \big] dv_g(z,y)
\end{align}
to get
\begin{equation}\label{def}
\partial_\eps \big( \E_{g_\eps} F_{g_\eps, \delta}(X) \big) = \hf \int_{\Sigma^2} \partial_\eps \big( G_{g_\eps}(z,y) \big) \E_{g_\eps} \big[ \tfrac{\delta^2 F_{g_\eps,\delta}(X)}{\delta X(z) \delta X(y)} \big] dv_g(z,y)+ \E_{g_\eps} [ \partial_\eps F_{g_\eps,\delta}(X)]\,.
\end{equation}
By a direct differentiation, the latter term on the right-hand side equals
\begin{align*}
\E_{g_\eps} [ \partial_\eps F_{g_\eps,\delta}(X)] &= - \tfrac{Q}{4\pi}   \int  \partial_\eps \big( K_{g_\eps}(z) \big) \E_{g_\eps} \big[ X(z) F_{g_\eps,\delta}(X) \big] dv_g(z)  \\
& \quad + \sum_j \alpha_j \int \E_{g_\eps} \big[ X(x') F_{g_\eps,\delta}(X) \big] \partial_\eps \nug{g_\eps}{x_j}(x') d^2 x' \\
& \quad - \mu \gamma \int \E_{g_\eps} \big[ X(z') \Vgd{g_\eps}{\gamma}(z) F_{g_\eps,\delta}(X) \big] \partial_\eps \nug{g_\eps}{z}(z') d^2z' dv_g(z)\,.
\end{align*}
Next, by applying Gaussian integration by parts (see Equation (9.1.32) in \cite{GJ}) to the expected values on the right-hand side, we get
\begin{align*}
\E_{g_\eps} [ \partial_\eps F_{g_\eps,\delta}(X)] &= - \tfrac{Q}{4\pi} \int  \partial_{\eps} K_{g_\eps}(z) G_{g_\eps}(z,y) \E_{g_\eps} \big[ \tfrac{\delta}{\delta X(y)} F_{g_\eps,\delta}(X) \big] dv_g(z) d^2y \\
& \quad + \sum_j \alpha_j \int G_{g_\eps}(x',y') \E_{g_\eps} \big[ \tfrac{\delta}{\delta X(y')} F_{g_\eps,\delta}(X) \big] \big(\partial_\eps \nug{g_\eps}{x_j}(x')\big) d^2x' d^2y' \\
& \quad - \mu \gamma \int G_{g_\eps}(z',y') \E_{g_\eps} \big[  \tfrac{\delta}{\delta X(y')} \big( \Vgd{g_\eps}{\gamma}(z) F_{g_\eps,\delta}(X) \big) \big] \big(\partial_\eps \nug{g_\eps}{z}(z')\big) d^2z' dv_g(z) d^2y'\,.
\end{align*}
Next we compute the functional derivatives
\begin{align*}
\tfrac{\delta}{\delta X(y)} \Vgd{g}{\alpha}(x) &= \tfrac{\delta}{\delta X(y)} \delta^{\frac{\alpha^2}{2}} e^{\alpha(c+X_g^\delta(x))} = \alpha \tfrac{\delta X_g^\delta (x)}{\delta X(y)} \Vgd{g}{\alpha}(x) = \alpha \nug{g}{x}(y) \Vgd{g}{\alpha}(x) \,, \\
\tfrac{\delta}{\delta X(y)} e^{- \tfrac{Q}{4\pi} \int K_g X_g^\delta dv_g} &= - \tfrac{Q}{4\pi} \tfrac{\delta}{\delta X(y)} \big( \int K_g X_g^\delta d v_g \big) e^{- \tfrac{Q}{4\pi} \int K_g X_g^\delta dv_g} = - \tfrac{Q}{4\pi} \int K_g (z) \nug{g}{z}(y) dv_g(z) e^{- \tfrac{Q}{4\pi} \int K_g X_g^\delta dv_g}\,, \\
\tfrac{\delta}{\delta X(y)} e^{- \mu \int \Vgd{g}{\gamma} dv_g} &= - \mu \tfrac{\delta}{\delta X(y)} \big( \int \Vgd{g}{\gamma}(z) dv_g(z) \big) e^{- \mu \int \Vgd{g}{\gamma} dv_g} = - \mu \gamma \int  \Vgd{g}{\gamma} (z) \nug{g}{z}(y) dv_g(z) e^{- \mu \int \Vgd{g}{\gamma} dv_g}\,.
\end{align*}

Now
\begin{align}
\frac{\delta}{\delta X(y)} F_{g,\delta}(X) &= \big( \sum_j \alpha_j \nug{g}{x_j}(y) - \tfrac{Q}{4\pi} \int K_g(z) \nug{g}{z}(y) dv_g(z) - \mu \gamma \int \Vgd{g}{\gamma}(z) \nug{g}{z}(y) dv_g(z) \big)   F_{g,\delta}(X)\,,
\end{align}
and
\begin{align}\label{ddF}
&\frac{\delta^2}{\delta X(y_1) X(y_2)} F_{g,\delta}(X) \\
&=  \prod_{k=1}^2  \Big(  \sum_j \alpha_j \nug{g}{x_j}(y_k) - \tfrac{Q}{4\pi} \int K_g(z_k) \nug{g}{z_k}(y_k) dv_g(z_k) - \mu \gamma \int \Vgd{g}{\gamma}(z_k) \nug{g}{z_k}(y_k) dv_g(z_k) \Big)  F_{g,\delta}(X) \nonumber \\
& \quad -\mu \gamma^2 \int  \Vgd{g}{\gamma}(z) \nug{g}{z}(y_1) \nug{g}{z}(y_2) dv_g(z)  F_{g,\delta}(X)\,. \nonumber
\end{align}

After plugging everything into \eqref{def} we get
\begin{align}\label{depsV}
&\partial_\eps \langle \prod_j \Vgd{g_\eps}{\alpha_j}(x_j) \rangle_{g_\eps} \\
&= \sum_{i,j} \tfrac{\alpha_i \alpha_j}{2} \partial_\eps \Big( \int G_{g_\eps}(y_1',y_2') d \nug{g_\eps}{x_i}(y_1') d \nug{g_\eps}{x_j}(y_2') \Big) \langle \prod_j \Vgd{g_\eps}{\alpha_j}(x_j) \rangle_{g_\eps} \nonumber \\
& \quad - \tfrac{Q}{4\pi} \sum_j \alpha_j \int_\Sigma \partial_\eps \Big( \int G_{g_\eps}(y_1',y_2') d \nug{g_\eps}{x_j}(y_1') d\nug{g_\eps}{z}(y_2') K_{g_\eps}(z)  \Big) dv_g(z) \langle \prod_j \Vgd{g_\eps}{\alpha_j}(x_j) \rangle_{g_\eps} \nonumber \\
& \quad - \mu \gamma \sum_j \alpha_j \int_\Sigma \partial_\eps \Big( \int G_{g_\eps}(y_1', y_2') d \nug{g_\eps}{x_j}(y_1') d \nug{g_\eps}{z}(y_2') \Big) \langle \Vgd{g_\eps}{\gamma}(z) \prod_j \Vgd{g_\eps}{\alpha_j}(x_j) \rangle_{g_\eps} dv_g(z) \nonumber \\
& \quad + \tfrac{Q^2}{2(4\pi)^2} \int_{\Sigma^2} \partial_\eps \Big( \int G_{g_\eps}(y_1',y_2') d\nug{g_\eps}{z_1}(y_1') d \nug{g_\eps}{z_2}(y_2') K_{g_\eps}(z_1) K_{g_\eps}(z_2)  \Big) dv_g(z_1,z_2) \langle \prod_j \Vgd{g_\eps}{\alpha_j}(x_j)\rangle_{g_\eps} \nonumber \\
& \quad + \tfrac{\mu \gamma Q}{4\pi} \int_{\Sigma^2} \partial_\eps \Big( \int G_{g_\eps}(y_1',y_2')  d\nug{g_\eps}{z_1}(y_1') d \nug{g_\eps}{z_2}(y_2') K_{g_\eps}(z_1)  \Big) \langle \Vgd{g_\eps}{\gamma}(z_2) \prod_j \Vgd{g_\eps}{\alpha_j}(x_j) \rangle_{g_\eps} dv_g(z_1,z_2) \nonumber \\
& \quad + \tfrac{\mu^2\gamma^2}{2} \int_{\Sigma^2} \partial_\eps \Big( \int G_{g_\eps}(y_1',y_2') d \nug{g_\eps}{z_1}(y_1') d \nug{g_\eps}{z_2}(y_2') \Big) \langle \Vgd{g_\eps}{\gamma}(z_1) \Vgd{g_\eps}{\gamma}(z_2) \prod_j \Vgd{g_\eps}{\alpha_j}(x_j) \rangle_{g_\eps} dv_g(z_1,z_2) \nonumber \\
& \quad - \tfrac{\mu \gamma^2}{2} \int_\Sigma  \partial_\eps \Big(\int G_{g_\eps}(y_1',y_2') d \nug{g_\eps}{z}(y_1') d \nug{g_\eps}{z}(y_2') \Big) \langle \Vgd{g_\eps}{\gamma}(z) \prod_j \Vgd{g_\eps}{\alpha_j}(x_j) \rangle_{g_\eps} dv_g(z)\,. \nonumber
\end{align}
The terms containing $\partial_\eps G_{g_\eps}$ come from the first term in \eqref{def} and terms containing $\partial_\eps \nu^\delta_{g_\eps}$ from the second term in \eqref{def}. With this formula it is easy to show the following

\begin{proposition}\label{se_tensor}
Let $g \in \meth$ and $f \in \sectt{}{}$. Define $g_\eps^{\alpha \beta} = g^{\alpha \beta} + \eps f^{\alpha \beta}$. Then the limit
\begin{align*}
\lim_{\delta \to 0}  \partial_\eps^n|_0 \langle \prod_j \Vgd{g_\eps}{\alpha_j}(x_j) \rangle_{g_\eps,\delta}
\end{align*}
exists.
\end{proposition}
\begin{proof}
First consider a case where the perturbations $g_\eps$ all are hyperbolic metrics (such perturbations are generated by $f \in \sectt{tf}{m}$). Then the scalar curvature $K_{g_\eps}$ stays constant and thus the derivatives $\partial_\eps K_{g_\eps}$ vanish. Then from Corollary \ref{corollary} and Equation \eqref{depsV} it follows that the derivative $\partial_\eps  \langle \prod_j \Vgd{g_\eps}{\alpha_j}(x_j) \rangle_{g_\eps}$ consist of terms of the form
\begin{align*}
\int_{\Sigma^2} \mathcal{F}^\delta_{g_\eps}(z_1,z_2;\bx) \langle \Vgd{g_\eps}{\gamma}(z_1) \Vgd{g_\eps}{\gamma}(z_2) \prod_j \Vgd{g_\eps}{\alpha_j}(x_j) \rangle_{g_\eps} dv_g(z_1,z_2)\,,
\end{align*}
(also terms with $0$ or $1$ $\Vgd{g_\eps}{\gamma}$-fields appear) where $\partial_\eps^k|_0 \mathcal{F}^\delta_{g_\eps}$ is uniformly bounded in $\delta$ and converges to a bounded limit as $\delta \to 0$. Thus, by iterating \eqref{depsV}, we see that the limit
\begin{align}
\lim_{\delta \to 0} \partial_\eps^n|_0 \langle \prod_j \Vgd{g_\eps}{\alpha_j}(x_j) \rangle_{g_\eps,\delta}
\end{align}
exists for all perturbations $f$ that keep the metric hyperbolic (i.e. $f \in \sectt{tf}{m}$).

For general perturbation $f \in \sectt{}{}$, we apply the Weyl anomaly formula. For any metric there exists a unique $\varphi$ such that $e^\varphi g$ is hyperbolic. Thus, there exists a family of functions $(\varphi_\eps)_{\eps \geq 0}$ depending smoothly on $\eps$ such that $g_\eps = e^{\varphi_\eps} h_\eps$ where $h_\eps \in \meth$. We get
\begin{align*}
\partial_\eps^n|_0 \langle \prod_j \Vgd{g_\eps}{\alpha_j}(x_j) \rangle_{g_\eps} &= \partial_\eps^n|_0 \Big( e^{cA(\varphi_\eps,h_\eps) - \sum_j \Delta_{\alpha_j} \varphi_\eps(x_j)} \langle \prod_j \Vgd{h_\eps}{\alpha_j}(x_j) \rangle_{h_\eps} \Big)\,.
\end{align*}
This derivative exists since from Corollary \ref{corollary} it follows that the $\eps$-derivatives of $\langle \prod_j \Vgd{h_\eps}{\alpha_j}(x_j) \rangle_{h_\eps}$ exist, and $\varphi_\eps$ depends smoothly on $\eps$.
\end{proof}

\subsection{Variations of the moduli}\label{moduli}

Let $f_k \in \sectt{tf}{m}$ for $k=1,\hdots,n$ and denote $\eps = (\eps_1,\hdots,\eps_n)$. Consider the perturbed metric
\begin{align}\label{vm}
g_\eps^{zz} = \sum_{k=1}^n \eps_k f_k^{zz}\,, \quad g_\eps^{z \bar z} = g^{z \bar z}\,.
\end{align}
We have shown that the partial derivative
\begin{align*}
\prod_{k=1}^n \partial_{\eps_k}|_0 \langle \prod_j \Vg{g}{\alpha_j}(x_j) \rangle_{g_\eps}
\end{align*}
exists. Since $\sectt{tf}{m}^{\otimes n}$ is finite-dimensional and $f_k \in \sect{tf}{m}$ in \eqref{vm}, we can write
\begin{align}\label{tmv}
(4\pi)^n \prod_{k=1}^n \partial_{\eps_k}|_0 \langle \prod_j \Vg{g}{\alpha_j}(x_j) \rangle_{g_\eps} &= \int_{\Sigma^n} \prod_{k=1}^n f_k^{z_k z_k}(z_k) \langle \prod_{k=1}^n T_m(z_k) \prod_j \Vg{g}{\alpha_j}(x_j) \rangle_g dv_g(\bz)\,,
\end{align}
where $\langle \prod_{k=1}^n T_m(z_k) \prod_j \Vg{g}{\alpha_j}(x_j) \rangle_g$ is a function satisfying
\begin{align}\label{Tm_orthog}
\int_\Sigma f^{z_i z_i}(z_i) \langle \prod_{k=1}^n T_m(z_k) \prod_j \Vg{g}{\alpha_j}(x_j) \rangle_g dv_g(z_i) = 0
\end{align}
for all $f \in \sectt{tf}{d}$. Concretely, we have
\begin{align*}
\langle \prod_{k=1}^n T_m(z_k) \prod_j \Vg{g}{\alpha_j}(x_j) \rangle_g &= (4\pi)^n \sum_{i_1,\hdots,i_n=1}^{3 \gen - 3} \prod_{l=1}^n \eta_{i_l, z_{i_l}, z_{i_l}}(z_{i_l}) \prod_{k=1}^n \partial_{\eps_k}|_0 \langle \prod_j \Vg{g_{i,\eps} }{\alpha_j}(x_j) \rangle_{g_{i,\eps}}\,,
\end{align*}
where $g_{i,\eps}^{zz} = \sum_{k=1}^n \eps_k \eta_{i_k}^{z_k z_k}$ and $(\eta_k)_{k=1}^{3\gen-3}$ is an orthonormal basis of $\sect{tf}{m}$. The object $T_m$ transforms as a tensor of order $(0,2)$, i.e. it has two lower indices $(T_m)_{\alpha \beta}$. We choose to not write these indices explicitly, because only the $(T_m)_{zz}$ will ever appear in our computations (the trace part $(T_m)_{z \bar z}$ vanishes, since all tensors $f \in \sectt{}{m}$ are traceless).

By using Equation \eqref{depsV} and Corollary \ref{corollary} we see that 
\begin{align}\label{Tm}
\langle \prod_{k=1}^n T_m(z_k) \prod_{j=1}^N \Vg{g}{\alpha_j}(x_j) \rangle_g &= \sum_{a=0}^{2n} \sum_{i_1,\hdots,i_n=1}^{3\gen-3}  \prod _{l=1}^n \eta_{i_l, z_{i_l} z_{i_l}}(z_{i_l}) \int_{\Sigma^a} F_{g,n}(\bx, \by) \langle \prod_{k=1}^a \Vg{g}{\gamma}(y_k) \prod_{j=1}^N \Vg{g}{\alpha_j}(x_j) \rangle_g dv_g(\by)
\end{align}
where $\by=(y_1,\hdots,y_a)$ and $F_{g,n}$ is bounded on $\Sigma^{a+N}$ and smooth in the region where all the $x_j$'s and $y_k$'s are disjoint. The derivative $\partial_{x_i} F_{g,n}$ behaves like $\frac{1}{x_i-y_k}$ as $x_i \to y_k$, which can be seen by studying the $\prod_k \partial_{\eps_k}|_0$ derivative of right-hand side of \eqref{lndgca} near the diagonal.

These properties enable us to show the following.
\begin{proposition}\label{Tm_smooth}
Denote $\bz = (z_1,\hdots,z_n)$ and $\bx = (x_1,\hdots,x_N)$. Then the function
\begin{align*}
(\bz,\bx) \mapsto \langle \prod_{k=1}^n T_m(z_k) \prod_{j=1}^N \Vg{g}{\alpha_j}(x_j) \rangle_g
\end{align*}
belongs to $C^\infty(\Sigma^n \times U^N)$, where
\begin{align*}
U_N &= \{ \bx = (x_1,\hdots,x_N) \subset \C^N : x_i \neq x_j \quad \forall i \neq j \}\,.
\end{align*}
\end{proposition}
\begin{proof}
In \cite{Oik} the case $n=0$ on the Riemann sphere was proven. Generalizing the $n=0$ case to arbitrary compact Riemann surface is straightforward, since the proof relies on local behaviour of the correlation functions and on the KPZ-identity (Lemma \ref{kpz}). What was actually shown in \cite{Oik} is that functions of the form
\begin{align*}
\bx \mapsto \int_{\Sigma^n} F(\bx,\bz) \langle \prod_{k=1}^n \Vg{g}{\gamma}(z_k) \prod_{j=1}^N \Vg{g}{\alpha_j}(x_j) \rangle_g dv_g(\bz)\,,
\end{align*}
where $F: \Sigma^N \times \Sigma^n \to \C$ is a function with certain prescribed singularities, are smooth on $U_N$. The integrals in Equation \eqref{Tm} fall into this category in the case when the $\partial_{x_i}$ derivaties don't hit the $F_{g,n}$ (since $F_{g,n}$ is bounded), and thus the method developed in \cite{Oik} is strong enough to prove smoothness of these terms. In the proof, one observes that the derivatives
\begin{align*}
\partial_{x_i} \langle \prod_{k=1}^a \Vg{g}{\gamma}(y_k) \prod_{j=1}^N \Vg{g}{\alpha_j}(x_j) \rangle_g
\end{align*}
will produce singular integrals of the form
\begin{align*}
\int \partial_{x_i} G_g(x_i,y_k) \langle \prod_{k=1}^{a+1} \Vg{g}{\gamma}(y_k) \prod_{j=1}^{N} \Vg{g}{\alpha_j}(x_j) \rangle_g dv_g(\by)\,, 
\end{align*}
which are not absolutely integrable, due to the combination of the  $\frac{1}{x_i-y_k}$ singularity of the term $\partial_{x_i} G_g(x_i,y_k)$ and the behavior of the correlation function as $x_i \to y_k$. One then has to regularize and use identities derived from Gaussian integration by parts to produce cancellations, which then yield a well-defined limit once the regularization is removed.

The case when we have $\partial_{x_i} F_{g,n}$ can essentially be reduced to the previous case, since the singularity of $\partial_{x_i} F_{g,n}$ is of the type $\frac{1}{x_i-y_k}$, so in the end these terms behave the same way as the terms where $\partial_{x_i}$ acts on the correlation function. Thus, we may assume that the $\partial_{x_i}$-derivative acts on the correlation function, since the produced singularities are identical. For completeness we sketch the main part of the smoothness argument presented in \cite{Oik}.

Application of Gaussian integration by parts implies that (see Section 3.2. of \cite{ward} for detailed discussion on this)
\begin{align}\label{asdf}
\partial_{x_i} \langle \prod_k \Vg{g}{\gamma}(y_k) \prod_j \Vg{g}{\alpha_j}(x_j) \rangle_g &= \alpha_i \gamma \sum_k \partial_{x_i} G_g(x_i,y_k) \langle \prod_k \Vg{g}{\gamma}(y_k) \prod_j \Vg{g}{\alpha_j}(x_j) \rangle_g \\
& \quad + \alpha_i \sum_{j \neq i} \alpha_j \partial_{x_i} G_g(x_i,x_j) \langle \prod_k \Vg{g}{\gamma}(y_k) \prod_j \Vg{g}{\alpha_j}(x_j) \rangle_g \nonumber \\
& \quad - \alpha_i \mu \gamma \int_\Sigma \partial_{x_i} G_g(x_i,z) \langle \Vg{g}{\gamma}(z) \prod_k \Vg{g}{\gamma}(y_k) \prod_j \Vg{g}{\alpha_j}(x_j) \rangle_g dv_g(z)\,.\nonumber
\end{align}
The integral on the last line is not absolutely convergent (see the Remark at the end of Section 3.2. in \cite{ward}). Thus, we need to show that a regularized version of this integral can be expressed in terms of integrals that, even after removing the regularization, converge. We do not explicitly include the regularization in our notation, since its presence does not affect the logic of the argument.

The key is to observe that the divergent integral above is almost like a Cauchy transform
\begin{align*}
\int \frac{1}{x-z} \langle \Vg{g}{\gamma}(z) \prod_k \Vg{g}{\gamma}(y_k) \prod_j \Vg{g}{\alpha_j}(x_j) \rangle_g dv_g(z)
\end{align*}
evaluated at $x=x_i$. If the above expression depends holomorphically on $x$, then we could ``guess"
\begin{align*}
&\int \frac{1}{x_i-z} \langle \Vg{g}{\gamma}(z)\prod_k \Vg{g}{\gamma}(y_k) \prod_j \Vg{g}{\alpha_j}(x_j) \rangle_g dv_g(z) \\
&= \frac{1}{2 \pi i} \oint_{|x-x_i|=r} \frac{1}{x-x_i} \Big( \int \frac{1}{x-z} \langle \Vg{g}{\gamma}(z)\prod_k \Vg{g}{\gamma}(y_k) \prod_j \Vg{g}{\alpha_j}(x_j) \rangle_g dv_g(z) \Big) dx\,.
\end{align*}
For a regularized version of the correlation function the above formula would hold, yet, it is not immediately clear how to compute the Cauchy transform. In \cite{ward} it was observed that if one instead studies the Beurling transform, where $(x-z)^{-1}$ is replaced by $(x-z)^{-2}$, one ends up with a useful identity. Stripping out the language of integral transforms, the computation boils down to the following observation. Using Gaussian integration by parts we get (we drop the $\prod_k \Vg{g}{\gamma}(y_k)$ from the notation for now)
\begin{align*}
  \partial_z \langle \Vg{g}{\gamma}(z) \prod_j \Vg{g}{\alpha_j}(x_j) \rangle_g  &= \sum_j \alpha_j \gamma \partial_z G_g(z,x_j) \langle \Vg{g}{\gamma}(z) \prod_j \Vg{g}{\alpha_j}(x_j) \rangle_g \\
 & \quad - \mu \gamma^2 \int \partial_z G_g(z,z') \langle \Vg{g}{\gamma}(z') \Vg{g}{\gamma}(z) \prod_j \Vg{g}{\alpha_j}(x_j) \rangle_g dv_g(z')\,.
\end{align*}
Now we integrate this equation in $z$ over a ball $B(x_i,r)$
\begin{align}\label{asdasd}
\int_{B(x_i,r)} \partial_z \langle \Vg{g}{\gamma}(z) \prod_j \Vg{g}{\alpha_j}(x_j) \rangle_g dv_g(z) &= \sum_j \alpha_j \gamma \int_{B(x_i,r)} \partial_z G_g(z,x_j) \langle \Vg{g}{\gamma}(z) \prod_j \Vg{g}{\alpha_j}(x_j) \rangle_g dv_g(z) \\
& \quad - \mu \gamma^2 \int_{B(x_i,r)} \int \partial_z G_g(z,z') \langle \Vg{g}{\gamma}(z') \Vg{g}{\gamma}(z) \prod_j \Vg{g}{\alpha_j}(x_j) \rangle_g dv_g(z,z')\,. \nonumber
\end{align}
In the sum the $j \neq i$ terms give absolutely convergent integrals. The $j=i$ term is essentially the problematic term we are trying to make sense of in Equation \eqref{asdf}. Thus, we need to analyze the left-hand side and the double integral on the right-hand side. In a conformal coordinate around $x_i$ we write $g = e^\sigma |dz|^2$, and integration by parts gives
\begin{align*}
\int_{B(x_i,r)} \partial_z \langle \Vg{g}{\gamma}(z) \prod_j \Vg{g}{\alpha_j}(x_j) \rangle_g dv_g(z) &= -\int_{B(x_i,r)} \partial_z \sigma(z) \langle \Vg{g}{\gamma}(z) \prod_j \Vg{g}{\alpha_j}(x_j) \rangle_g dv_g(z) \\
& \quad + \frac{i}{2} \int_{\partial B(x_i,r)} \langle \Vg{g}{\gamma}(z) \prod_j \Vg{g}{\alpha_j}(x_j) e^{\sigma(z)} d \bar z\,.
\end{align*}
Both of the integrals are absolutely convergent (we take $r$ small enough so that $\partial B(x_i,r)$  does not touch the points $\{x_j\}_j$). The double integral in \eqref{asdasd} can be written as
\begin{align}\label{asd}
& \int_{B(x_i,r)} \int \partial_z G_g(z,z') \langle \Vg{g}{\gamma}(z') \Vg{g}{\gamma}(z) \prod_j \Vg{g}{\alpha_j}(x_j) \rangle_g dv_g(z,z') \\
&= \int_{B(x_i,r)^2} \partial_z G_g(z,z') \langle \Vg{g}{\gamma}(z') \Vg{g}{\gamma}(z) \prod_j \Vg{g}{\alpha_j}(x_j) \rangle_g dv_g(z,z')  \nonumber\\
& \quad + \int \mathcal{F}_g(x_i,r;z,z') \langle \Vg{g}{\gamma}(z') \Vg{g}{\gamma}(z) \prod_j \Vg{g}{\alpha_j}(x_j) \rangle_g dv_g(z,z')\,, \nonumber
\end{align}
where 
\begin{align*}
\mathcal{F}_g(x_i,r;z,z') &= \mathbf{1}_{B(x_i,r)}(z) \mathbf{1}_{B(x_i,r)^c}(z') \partial_z G_g(z,z')\,.
\end{align*}
This is the type of the described singularity we mentioned earlier. The first term on the right-hand side of \eqref{asd} is absolutely convergent after we observe that the singular part of $\partial G_g(z,z')$ cancels by antisymmetry in $z$ and $z'$. What is left is to analyze the integrability of the second term on the right-hand side of \eqref{asd}. The fusion estimates derived in Proposition 5.1. of \cite{ward} tell us the singular behavior of $\langle \Vg{g}{\gamma}(z') \Vg{g}{\gamma}(z) \prod_j \Vg{g}{\alpha_j}(x_j) \rangle_g$ as $z' \to z$, and then it is simple to check that the resulting integral converges near $z=z'$. Now from Equation \eqref{asdasd} we can solve for the singular part of the problematic integral term in \eqref{asdf}, and we see that the result is a sum of well-defined terms.

For the higher derivatives $\partial_{x_i}^n$, one ends up with a sum of integrals of the type
\begin{align*}
\int \prod_{k=1}^n \mathcal{F}_g(x_i,2^{-k} r; z_k,z'_k) \langle \prod_{k=1}^n \Vg{g}{\gamma}(z_k) \Vg{g}{\gamma}(z_k') \Vg{g}{\alpha_j}(x_j) \rangle_g dv_g(\bz,\bz')\,,
\end{align*}
and less singular terms. Then, Lemma 3.1. in \cite{Oik} implies the convergence of this integral.

When we apply the formula for $\partial_{x_i}^n \langle \prod_{k=1}^a \Vg{g}{\gamma}(y_k) \prod_{j=1}^N \Vg{g}{\alpha_j}(x_j) \rangle_g$ for the $\partial_{x_i}^n$ derivatives of the right-hand side of Equation \eqref{Tm}, we get integrals with at worst behave as
\begin{align*}
\int_{\Sigma^{a+2n}} F_{g,n}(\bx,\by) \prod_{l=1}^n \mathcal{F}_g(x_i, 2^{-l}r,w_l,w_l') \langle \prod_{l=1}^n \Vg{g}{\gamma}(w_l) \Vg{g}{\gamma}(w_l') \prod_{k=1}^a \Vg{g}{\gamma}(y_k) \prod_{j=1}^N \Vg{g}{\alpha_j}(x_j) \rangle_g dv_g(\by,\mathbf{w}, \mathbf{w'})\,,
\end{align*}
and the integrand has an integrable dominant according to Lemma 3.1. in \cite{Oik}. The result follows, since the integrand in \eqref{Tm} is differentiable almost everywhere on the integration domain and the derivative has an integrable dominant.
\end{proof}

\section{Conformal Ward identities}\label{section_ward}

In this Section we show that the SE-tensor can be expressed in terms of a pointwise defined function. After this we derive the conformal Ward identities for the SE-tensor.

\subsection{Beltrami equation}\label{beltrami_section}

In this section we assume that $g$ is a hyperbolic metric. For a general variation $g_\eps$ of $g$ we can write
\begin{align}\label{metric_decomp}
g_\eps &= e^{\varphi_\eps}\psi_\eps^* h_{\eps}
\end{align}
where $\varphi_\eps \in C^\infty(\Sigma)$, $\psi_\eps \in \mathcal{D}(\Sigma)$ and $h_\eps$ is a hyperbolic metric, chosen to represent some equivalence class in the moduli space, for each $\eps$.
For now we consider a variation that in terms of the inverse metric is of the form $g_\eps^{\alpha \beta} = g^{\alpha \beta} + f_\eps^{\alpha \beta}$, where $\eps = (\eps_1,\hdots,\eps_n) \in \C^n$, and 
\begin{align*}
g_\eps^{zz} &= f_\eps^{zz} = \sum_{k=1}^n \eps_k f_k^{zz}\,, \\
g_\eps^{\bar z \bar z} &= \overline{g_\eps^{zz}}\,, \\
g_\eps^{z \bar z} &= g^{z \bar z}\,,
\end{align*}
with $f_k \in \sectt{tf}{}$, $k=1,\hdots,n$, having mutually disjoint supports. Equation \eqref{metric_decomp} implies that $\psi_\eps$ solves the non-linear Beltrami equation\footnote{In the case of the sphere the equation degenerates to the linear Beltrami equation, because $\nu_\eps$ vanishes if the moduli space is trivial. This is why it sufficed to study the linear case in \cite{KO}.} (see Equations (10.4)-(10.6) in Section 10.1. of \cite{astala})
\begin{align}\label{beltr0}
\partial_{\bar z} \psi_\eps &= (\mu_\eps)^z_{\bar z} \partial_z \psi_\eps - (\nu_\eps)^z_{\bar z} \overline{\partial_z \psi_\eps}\,,
\end{align}
where
\begin{align*}
(\mu_\eps)^z_{\bar z} &= \frac{\tilde g_{\eps, \bar z \bar z}}{\tilde g_{\eps, z \bar z} + \tilde h_{\eps, z \bar z}}\,, \quad (\nu_\eps)^z_{\bar z} = \frac{\tilde h_{\eps, \bar z \bar z}}{\tilde g_{\eps, z \bar z} + \tilde h_{\eps, z \bar z}}\,, \\
\tilde g_\eps &= \frac{g_\eps}{\sqrt{\det g_\eps}} \,, \quad \tilde h_\eps = \frac{h_\eps \circ \psi_\eps }{\sqrt{\det (h_\eps \circ \psi_\eps)}}\,.
\end{align*}
Let $D \psi_\eps$ be the tangent map $\Gamma(T^1\Sigma) \to \Gamma(T^1\Sigma)$ of $\psi_\eps$. In terms of this, Equation \eqref{beltr0} becomes
\begin{align}\label{beltr}
(D\psi_\eps)^{zz} &= (\mu_\eps)^z_{\bar z} (D\psi_\eps)^{\bar z z} - (\nu_\eps)^z_{\bar z}  (D \psi_\eps)^{z \bar z}\,.
\end{align}
%Note that for traceless $f_i,f_j$, we have $\partial_{\eps_i}|_0 \partial_{\eps_j}|_0 \mu_\eps = 0 = \partial_{\eps_i}|_0 \partial_{\eps_j}|_0 \nu_\eps$, i.e. the higher order $\eps$-dependency is in the mixed powers $\eps_i^a \overline{\eps_j}^b$. 
For the perturbations we have the orthogonal decomposition $f_k = f_{k,d} + f_{k,m}$, where $f_{k,d} \in \sect{tf}{d}$ and $f_{k,m} \in \sect{tf}{m}$ (Lemma \ref{tangent_decomp}). Since $f_{k,m}$ is the part generating the deformation of the conformal structure, the hyperbolic metrics $h_\eps$ in \eqref{metric_decomp} depend only on this component of $f_k$. This means that we have $h_\eps^{zz} = \sum_k \eps_k f_{k,m}^{zz} + \mathcal{O}(\eps^2)$, $h_\eps^{z \bar z} = g^{z \bar z} + \mathcal{O}(\eps^2)$, from which we get
\begin{align*}
  (\dot \mu_i)^{z}_{\bar z} &:= \partial_{\eps_i}|_0 (\mu_\eps)^{z}_{\bar z} = - \hf g^{z \bar z} f_{i,\bar z \bar z} = - \hf g_{z \bar z}  f_i^{zz}\,, \\
 ( \dot \nu_i)^{z}_{\bar z} &:= \partial_{\eps_i}|_0 (\nu_\eps)^z_{\bar z} = -\hf g_{z \bar z} f_{i,m}^{zz}\,.
\end{align*}

To show smoothness of $\psi$ in $(\eps,z)$, we can reduce \eqref{beltr} to a pair of ordinary Beltrami equations. Because the metrics $g$ and $h_\eps$ are hyperbolic, we can find a locally defined diffeomorphism $\psi_{2,\eps}$ such that
\begin{align}\label{bel1}
h_\eps = \psi_{2,\eps}^* g\,.
\end{align}
Define $\psi_{1,\eps} := \psi_{2,\eps} \circ \psi_\eps$. Then 
\begin{align}\label{bel2}
g_\eps &= e^{\varphi_\eps } \psi_{1,\eps}^* g\,.
\end{align}
Equations \eqref{bel1} and \eqref{bel2} imply that $\psi_{1,\eps}$ and $\psi_{2,\eps}$ solve the standard Beltrami equations
\begin{align*}
\partial_{\bar z} \psi_{2,\eps} &= \hat \nu_\eps \partial_z \psi_{2,\eps}\,, \\
\partial_{\bar z} \psi_{1,\eps} &= \hat \mu_\eps \partial_z \psi_{1,\eps}\,,
\end{align*}
where
\begin{align*}
\hat \nu_\eps &= \frac{\tilde h_{\eps, \bar z \bar z}}{\tilde g_{z \bar z} + \tilde h_{\eps,z \bar z}}\,, \quad \hat \mu_\eps = \frac{\tilde g_{\eps, \bar z \bar z}}{\tilde g_{z \bar z} + \tilde g_{\eps, z \bar z}}\,.
\end{align*}
Thus, the function $\psi_\eps$ can always be locally decomposed as $\psi_\eps = \psi_{2,\eps}^{-1} \circ \psi_{1,\eps}$, where $\psi_{1,\eps}$ and $\psi_{2,\eps}$ are solutions of standard Beltrami equations.
\begin{proposition}
The map $\psi: \Sigma \times \C \to \Sigma$, $\psi(z,\eps) = \psi_\eps(z)$, is smooth.
\end{proposition}
\begin{proof}
It is enough to fix a point $x \in \Sigma$ and study the maps $\psi_1$ and $\psi_2$, satisfying $\psi = \psi_2^{-1} \circ \psi_1$, defined in some neighbourhood of $x$. Let $(U_n, \rho_n)_n$ be a partition of unity, and define
\begin{align*}
\hat \mu_{\eps,n} &:= \rho_n \hat \mu_{\eps}\,,
\end{align*}
where $n$ is such that $x \in U_n$ and $\rho_n = 1$ around $x$. We will apply Stoilow factorization (Theorem 5.5.1. in \cite{astala}). Let $\psi_{1,\eps}^{(n)}$ be a solution to $\partial_{\bar z} \psi_{1,\eps}^{(n)} = \hat \mu_{\eps,n} \partial_z \psi_{1,\eps}^{(n)}$. Because $\hat \mu_{\eps,n} = \hat \mu_\eps$ around $x$, Stoilow factorization says that there exists a conformal map $\xi_\eps$ such that $\xi_\eps \circ \psi_{1,\eps}^{(n)} =  \psi_{1,\eps}$ in some neighbourhood of $x$. Now, by the argument in Proposition 3.1. of \cite{KO} (which we can now use since we have localized everything to planar domains), $\psi_1^{(n)}$ is jointly smooth in $(\eps,z)$, and since $\xi_\eps \circ \psi_{1,\eps}^{(n)}$ solves the same Beltrami equation, the same argument shows that also $\xi \circ \psi_{1}^{(n)}$ is jointly smooth in $(z,\eps)$. Thus $\psi_1$ is jointly smooth in $(z,\eps)$.

Smoothness of $\psi_2^{-1}$ follows similarly, since by Theorem 5.5.6. in \cite{astala} $\psi_2^{-1}$ also solves a standard Beltrami equation. Thus the same argument works for $\psi_2^{-1}$.
\end{proof}

Next we  compute the partial derivatives $\prod_k \partial_{\eps_k}|_0 \psi_\eps$ and $\prod_k \partial_{\eps_k}|_0 \varphi_\eps$. Denote $u_i = \partial_{\eps_i}|_0 \psi_\eps$. By taking $\partial_{\eps_i}|_0$ of \eqref{beltr} and using $\nabla^z = g^{z \bar z} \partial_{\bar z}$ we get
\begin{align*}
\nabla^z u_i^z &=  ( (\dot\mu_i)^z_{\bar z} -  (\dot \nu_i)^z_{\bar z})g^{\bar z  z} = - \hf (f_i^{zz} - f_{i,m}^{zz})= -\tfrac{1}{2}  f_{i,d}^{zz}\,,
\end{align*}
which implies that, written in terms of the inverse $\mathcal{G}$ of the conformal Killing operator $P_g$ introduced in Section \ref{conformal_killing}, we get $u_i = - \mathcal{G} f_{i,d} = -  \mathcal{G} f_i$.\footnote{In \cite{KO} the role of $\mathcal{G}$ was played by Cauchy transform. One should think that $\mathcal{G}$ maps the perturbation $f$ to the vector field $u$ for which $g+\eps f_d = g + \eps \mathcal{L}_u g$.}

Next we compute the derivative $\partial_{\eps_i} \partial_{\eps_j}|_0$ of \eqref{beltr}. Most of the $\eps^2$-dependency is in mixed powers $\eps_i \overline{\eps_j}$, so we do not get many terms. We have 
\begin{align*}
&\partial_{\eps_i} \partial_{\eps_j}|_0 \sqrt{\det g_\eps} = \partial_{\eps_i} \partial_{\eps_j}|_0 \sqrt{\det h_\eps } = 0\,,\\
& \partial_{\eps_i}|_0 \partial_{\eps_j}|_0 g_{\eps, \bar z \bar z} = \partial_{\eps_i}|_0 \partial_{\eps_j}|_0 h_{\eps, \bar z \bar z} = 0 \\
&(\nu_\eps)^z_{\bar z} = e^{- \sigma \circ \psi_\eps} h_{\eps, \bar z \bar z} \circ \psi_\eps + \mathcal{O}(\eps^2) = - \tfrac{1}{4} (g_{z \bar z} \circ \psi_\eps ) (h_\eps^{zz} \circ \psi_\eps) + \mathcal{O}(\eps^2)\,.
\end{align*}
In the last equation the minus sign comes from computing the inverse matrix of $h_\eps$, because $h_\eps^{\alpha \beta}$ denotes the inverse of $h_\eps$ and not just $h_{\eps,\alpha \beta}$ with indices raised. It follows that 
\begin{align}
\partial_{\eps_i}|_0 \partial_{\eps_j}|_0 (\nu_\eps)^{z}_{\bar z}  &= u_i^z \nabla_z  (\dot \nu_j)^{z}_{\bar z} + u_j^z \nabla_z (\dot \nu_i)^{z}_{\bar z} = 0\,,
\end{align}
where in the last equality we used $(\dot \nu_i)^{z}_{\bar z} = - \hf g_{z \bar z} f_{i,m}^{zz}$ and $\nabla_z f_{i,m}^{zz} = 0$. Now for for $ u_{ij}^z(z) := \partial_{\eps_i} \partial_{\eps_j} \psi_\eps(z)|_0$ we get the equation
\begin{align*}
(D u_{ij})^{zz} &= \dot \mu_{i,\bar z}^z (D u_j)^{\bar z z} + \dot \mu_{j,\bar z}^z (D u_i)^{\bar z z} 
\end{align*}
which in turn implies that $u_{ij} = - 2\mathcal{G} (\dot \mu_i \mathcal{B} f_j) - 2 \mathcal{G}(\dot \mu_j \mathcal{B} f_i)$, where\footnote{The notation $\mathcal{B}$ is chosen to be analogous to the integral transform $\mathcal{B} = \partial_z \mathcal{C}$ used in \cite{KO}.}
\begin{align*}
(\mathcal{B} f)^{\bar z z} &:= \nabla^{\bar z} (\mathcal{G} f)^z\,.
\end{align*}
In general for $I \subset \{1,\hdots,n\}$ we denote  $\prod_{i \in I} \partial_{\eps_i} \psi_\eps|_0 = u_I^z$ and one obtains
\begin{align}\label{u_I}
u_I^z &= (-1)^{|I|+1}\sum_{\pi \in S(I)} 2 \big(\mathcal{G} \dot \mu_{\pi(1)} \mathcal{B}  f_{\pi(2)} \hdots \mathcal{B} f_{\pi(|I|)}\big)^z \,.
\end{align}
For $\varphi_\eps$ we get the equation (by computing similarly as in Section 3.1. of \cite{KO})
\begin{align*}
\varphi_\eps(z) &= -  \log \det D \psi_\eps - \hf \log \frac{\det h_\eps \circ \psi_\eps }{\det g_\eps}\,.
\end{align*}
Since $f_k^{z \bar z} = 0$ for all $k$, it follows that $\prod_{k \in I} \partial_{\eps_k} \log  \det g_\eps|_{\eps=0} = 0$ for all $I \subset \{1,\hdots,k\}$, and that 
\begin{align*}
\prod_{k\in I} \partial_{\eps_k} \log \det h_\eps \circ \psi_\eps |_{\eps=0}  &= \prod_{k \in I} \partial_{\eps_k} (2 \sigma \circ \psi_\eps) |_{\eps=0}\,.
\end{align*}
Thus,
\begin{align}\label{phi_I}
\dot \varphi_I(z) &= -  \prod_{i \in I} \partial_{\eps_i}|_0  \log \nabla_z \psi_\eps  = \sum_{P \in \mathcal{P}(I)} C_P \prod_{J \in P} \nabla_z u_J^z
\end{align}
where $\mathcal{P}(I)$ is the set of partitions of $I$ and $C_P$ are some constants with $C_P = 1$ when $|I|=1$.

\subsection{Stress-Energy Tensor}
Recall from Section \ref{moduli} that we denote by $\langle \prod_i T_m(z_i) \prod_j \Vg{g}{\alpha_j}(x_j) \rangle_g$ the pointwise defined function that describes the modular variation of the correlation function (Equation \eqref{tmv}). This and the transformation properties of the correlation functions imply the following
\begin{proposition}\label{T_m_transformation}
Let $(x_1,\hdots,x_N) \in \Sigma^N$ and $(z_1,\hdots,z_n) \in (\Sigma \setminus \{x_j\}_j)^n$ be tuples of disjoint points. Then for all $\psi \in \dif$ 
\begin{align}\label{tm_diffeo}
\langle \prod_{k=1}^n T_m(z_k) \prod_{j=1}^N \Vg{\psi^* g}{\alpha_j}(x_j) \rangle_{\psi^* g} &= \langle \prod_{k=1}^n (\psi^* T_m)(z_k) \prod_{j=1}^N \Vg{g}{\alpha_j}(\psi(x_j)) \rangle_g\,,
\end{align}
where $\psi^* T_m = (D^T\psi)_z^{\; \; \alpha} (T_m \circ \psi)_{\alpha \beta} (D \psi)^\beta_{\;\;z}$, and
for all $\varphi \in C^\infty(\Sigma)$
\begin{align}\label{tm_weyl}
\langle \prod_{k=1}^n T_m(z_k) \prod_{j=1}^N \Vg{e^\varphi g}{\alpha_j}(x_j) \rangle_{e^\varphi g} &= e^{cA(\varphi,g)- \sum_{j=1}^N \Delta_{\alpha_j} \varphi(x_j)} \langle \prod_{k=1}^n \big( T_m(z_k) + a(z_k,\varphi,g)\big) \prod_{j=1}^N \Vg{g}{\alpha_j}(x_j) \rangle_g\,,
\end{align}
where $a(z,\varphi,g) = 4 \pi c \frac{\delta A(\varphi,g)}{\delta g^{zz}(z)}$.
\end{proposition}
\begin{remark} The trace part of $T_m$ vanishes, i.e. $(T_m)_{z \bar z} = 0$, because the perturbations $f$ that deform the conformal structure are traceless. Thus, only the $(z,z)$ and $(\bar z, \bar z)$ components of $T_m$ appear in Equation \eqref{tm_diffeo}.
\end{remark}
\begin{proof}
\begin{enumerate}
\item By definition
\begin{align*}
\int_{\Sigma^n} \prod_{i=1}^n f_i^{z_i z_i}(z_i) \langle \prod_{i=1}^n T_m(z_i) \prod_j \Vg{g}{\alpha_j}(x_j) \rangle_{\psi^* g} dv_{\psi^* g}(\bz) &= (4\pi)^n \prod_{i=1}^n \partial_{\eps_i}|_0 \langle \prod_j \Vg{g}{\alpha_j}(x_j) \rangle_{g_\eps}
\end{align*}
where $g_\eps =  \psi^* g + \sum_i \eps_i f_{i,m} + \mathcal{O}(\eps^2) = \psi^* (g + \sum_i \eps_i \psi_* f_{i,m}) + \mathcal{O}(\eps^2) =: \psi^* \tilde g_\eps$. Thus
\begin{align*}
(4\pi)^n \prod_{i=1}^n \partial_{\eps_i}|_0 \langle \prod_j \Vg{g}{\alpha_j}(x_j) \rangle_{g_\eps} &= (4\pi)^n \prod_{i=1}^n \partial_{\eps_i}|_0 \langle \prod_j \Vg{g}{\alpha_j}(\psi(x_j)) \rangle_{\tilde g_\eps} \\
&=  \int_{\Sigma^n} \prod_{i=1}^n (\psi_* f_{i,m})^{z_i z_i}(z_i) \langle \prod_{i=1}^n T_m(z_i) \prod_j \Vg{g}{\alpha_j}(\psi(x_j)) \rangle_g dv_g(\bz) \\
&=  \int_{\Sigma^n} \prod_{i=1}^n f_{i,m}^{z_i z_i}(z_i) \langle \prod_{i=1}^n (\psi^* T_m)(z_i) \prod_j \Vg{g}{\alpha_j}(\psi(x_j)) \rangle_g dv_{\psi^* g}(\bz)\,.
\end{align*}
Above in the second equality we used that if $f \in T_{\psi^*g}^{\opn{tt}} \met$, then $\psi_* f \in \trtr$ (recall Definition \ref{transv1}, Corollary \ref{transv2}), which means that $\psi_* f_{i,m}$ generates a perturbation of $g$ that is purely a variation of the moduli (and not a pull-back by a diffeomorphism or a Weyl transformation). This allows us to write the partial derivative in terms of the pointwise defined function $\langle \prod_i T_m(z_i) \prod_j \Vg{g}{\alpha_j}(\psi(x_j)) \rangle_g$ defined in Equation \eqref{tmv}. 
 %This holds because if $\{\eta_k\}_{k=1}^{3\gen-3}$ is an orthonormal basis for $\sect{tf}{m}$ in the inner product $(\cdot,\cdot)_{\psi^*g}$, then $\{\psi_* \eta_k\}_{k=1}^{3\gen-3}$ is an orthonormal basis of $\opn{QD}_g(\Sigma)$, and we have that $\langle \psi_* f, \psi_* \eta_k \rangle_{g} = \langle f, \eta_k \rangle_{\psi_* g}$.
\item By definition
\begin{align*}
\int_{\Sigma^n} \prod_{i=1}^n f_i^{z_i z_i} \langle \prod_{i=1}^n T_m(z_i) \prod_j \Vg{e^\varphi g}{\alpha_j}(x_j) \rangle_{e^\varphi g} dv_{e^\varphi g}(\bz) &= (4\pi)^n \prod_{i=1}^n \partial_{\eps_i}|_0 \langle \prod_j \Vg{g_\eps}{\alpha_j}(x_j) \rangle_{g_\eps}\,,
\end{align*}
where $g_\eps = e^\varphi g + \sum_i \eps_i f_{i,m} = e^\varphi(g+\sum_i \eps_i e^{-\varphi} f_{i,m})+ \mathcal{O}(\eps^2) =: e^\varphi \tilde g_\eps$. Thus
\begin{align*}
\prod_{i=1}^n \partial_{\eps_i}|_0 \langle \prod_j \Vg{g_\eps}{\alpha_j}(x_j) \rangle_{g_\eps} &= \prod_{i=1}^n \partial_{\eps_i}|_0 e^{cA(\varphi, \tilde g_\eps) - \sum_j \Delta_{\alpha_j} \varphi(x_j)} \langle \prod_j \Vg{\tilde g_\eps}{\alpha_j}(x_j) \rangle_{\tilde g_\eps}\,.
\end{align*}
We denote $a_{\alpha \beta}(z,\varphi,g) = \frac{\delta A(\varphi,g)}{\delta g^{\alpha \beta}(z)}$. Then, by definition of the functional derivative,
\begin{align*}
\partial_{\eps_i}|_0 A(\varphi,\tilde g_\eps) &=  \int_\Sigma e^{\varphi(z)} f_{i,m}^{zz}(z) a_{zz}(z,\varphi,g) d^2z\,.
\end{align*}
Combining this with Equation \eqref{tmv} yields
\begin{align*}
& (4\pi)^n \prod_{i=1}^n \partial_{\eps_i}|_0 e^{cA(\varphi, \tilde g_\eps) - \sum_j \Delta_{\alpha_j} \varphi(x_j)} \langle \prod_j \Vg{\tilde g_\eps}{\alpha_j}(x_j) \rangle_{\tilde g_\eps} \\
&= e^{cA(\varphi, g) - \sum_j \Delta_{\alpha_j} \varphi(x_j)} \int_{\Sigma^n} \prod_{i=1}^n e^{\varphi(z_i)} f_{i,m}^{z_i z_i}(z_i)\langle \prod_{i=1}^n \big( T_m(z_i) + a(z_i,\varphi,g) \big) \prod_j \Vg{g}{\alpha_j}(x_j) \rangle_g dv_g(\bz) \\
&= e^{cA(\varphi, g) - \sum_j \Delta_{\alpha_j} \varphi(x_j)} \int_{\Sigma^n} \prod_{i=1}^n  f_{i,m}^{z_i z_i}(z_i)\langle \prod_{i=1}^n \big( T_m(z_i) + a(z_i,\varphi,g) \big) \prod_j \Vg{g}{\alpha_j}(x_j) \rangle_g dv_{e^\varphi g}(\bz)\,,
\end{align*}
where we used that if $f \in T_{e^\varphi g}^{\opn{tt}} \met$, then $e^{\varphi} f \in \trtr$, to make sure that $e^\varphi f$ generates a perturbation of $g$ that is purely a variation of the moduli (so that we can use Equation \eqref{tmv}).
\end{enumerate}
\end{proof}

Now we are ready to show that the SE-tensor is given by a pointwise defined function for mutually disjoint perturbations.
\begin{proposition}\label{T_pointwise}
Let $f_k \in \sectt{tf}{}$ for $k=1,\hdots,n$ be symmetric tensors with mutually disjoint supports and denote $\eps = (\eps_1,\hdots,\eps_n)$. Define
\begin{align*}
g_\eps^{zz} &= \sum_{k=1}^n \eps_k f_k^{zz}\,, \\ 
g_\eps^{z \bar z} &= g^{z \bar z}\,, \\
g_{\eps}^{\bar z \bar z} &= \overline{g_\eps^{zz}}\,, \;  g_\eps^{ \bar z z} = g_\eps^{ z \bar z}\,.
\end{align*}
Then there exists functions $\langle \prod_{k=1}^n T_{z_k z_k}(z_k) \prod_j \Vg{g}{\alpha_j}(x_j)\rangle_g$ satisfying
\begin{align}\label{tpw}
(4\pi)^n \prod_{k=1}^n \partial_{\eps_k}|_0 \langle \prod_j \Vg{g}{\alpha_j}(x_j) \rangle_{g_\eps} &= \int_{\Sigma^n} \prod_{k=1}^n f_k^{z_k z_k}(z_k)\langle \prod_{k=1}^n T_{z_k z_k}(z_k) \prod_j \Vg{g}{\alpha_j}(x_j) \rangle_g dv_g(\bz)\,.
\end{align}
The functions $\langle \prod_{k=1}^n T_{z_k z_k}(z_k) \prod_j \Vg{g}{\alpha_j}(x_j)\rangle_g$ are smooth in the region of non-coinciding points.
\end{proposition}

\begin{proof} By definition we have that
\begin{align}\label{zxc}
&\prod_{k=1}^n \partial_{\eps_k}|_0 \langle \prod_j \Vg{g}{\alpha_j}(x_j) \rangle_{g_\eps} \\
&= \prod_{k=1}^n \partial_{\eps_k}|_0\Big( e^{c A(\varphi_\eps,\psi_\eps^* h_\eps) - \sum_j \Delta_{\alpha_j} \varphi_\eps(x_j)} \langle \prod_j \Vg{g}{\alpha_j}(\psi_\eps(x_j)) \rangle_{h_\eps} \Big) \nonumber\\
&= \sum_{I,J,K \subset \mathcal{P}(n)} \partial_{\eps_I}|_0 \Big( e^{c A(\varphi_\eps,\psi_\eps^* h_\eps) - \sum_j \Delta_{\alpha_j} \varphi_\eps(x_j)} \Big) \tfrac{1}{(4\pi)^{|K|}} \partial_{\eps_J}|_0 \int_{\Sigma^{|K|}} \prod_{k \in K} f_{k,m}(z_k) \langle \prod_{k \in K} T_m(z_k) \prod_j \Vg{g}{\alpha_j}(\psi_\eps(x_j)) \rangle_{g} dv_g(\bz)\,,\nonumber
\end{align}
where the sum is over partitions $\{I,J,K\}$ of $\{1,\hdots,n\}$ and we denote $\partial_{\eps_I} = \prod_{i \in I} \partial_{\eps_i}$. Above we split the $\eps_i$-derivatives into 3 groups $I$,$J$ and $K$, where $I$ is for derivatives acting on the Weyl anomaly, $J$ for derivatives acting on the diffeomorphism $\psi_\eps$, and $K$ for derivatives acting on the $h_\eps$ (variations of the moduli).

Using Equations \eqref{phi_I} and \eqref{u_I} we get
\begin{align}\label{t1}
\partial_{\eps_I}|_0 \Big( e^{ - \sum_j \Delta_{\alpha_j} \varphi_\eps(x_j)} \Big) &= \int \prod_{i \in I} f_i^{z_i z_i}(z_i) F_I(\bx,z_I) dv_g(z_I)\,,
\end{align}
where $F_I$ is smooth in the region of non-coinciding points. To see that the contribution of $c A(\varphi_\eps,\psi_\eps^* h_\eps)$ is of the wanted form \eqref{tpw}, a small computation is required. We first write
\begin{align}\label{Adm}
A(\varphi_\eps, \psi_\eps^* h_\eps) &= A(\varphi_\eps \circ \psi_\eps^{-1}, h_\eps) = \frac{1}{96 \pi} \int_\Sigma \big( h_\eps^{\alpha \beta} \partial_\alpha (\varphi_\eps \circ \psi_\eps^{-1}) \partial_\beta (\varphi_\eps \circ \psi_\eps^{-1}) + 2 K_{h_\eps} \varphi_\eps \circ \psi_\eps^{-1} \big) dv_g\,.
\end{align}
We have $K_{h_\eps} = -2$ for all $\eps \geq 0$, thus the second term is simple to deal with. We have to be careful with the first term, since $\partial_{\eps_1}|_0 h_\eps^{zz} = f_{1,m}^{zz}$ and $\partial_{\eps_1}|_0 \partial_{\bar z} \psi_\eps = g_{z \bar z} \nabla^z (\mathcal{G} f_1)^z = -\hf g_{z \bar z} f_{1,d}^{zz}$, and we want the end result to be of the form \eqref{tpw}, where only $f_1$ appears. Thus, we have to check that these terms always occur as a sum $f_{1,m}+f_{1,d}=f_1$.\footnote{The reason we want the whole perturbations $f_i$ to appear in \eqref{tpw} and not just the components $f_{i,d}$ and $f_{i,m}$ is that we assumed the supports of $f_i$ to be disjoint, but even with this assumption, we have no knowledge of the supports of $f_{i,d}$ and $f_{i,m}$.}

We split $h_\eps^{\alpha \beta} = g^{\alpha \beta} + \sum_{k=1}^n \eps_k f_{k,m}^{\alpha \beta}$ in \eqref{Adm} and look at each of the terms separately. Denote $\xi_\eps := \psi_\eps^{-1}$. For the $\partial_{\eps_I}|_0$-derivative of the first term we get
\begin{align}\label{qwe1}
2 \int_\Sigma g^{z \bar z} \prod_{i \in I} \partial_{\eps_i}|_0 \big(  \partial_{\bar z} (\varphi_\eps \circ \xi_\eps) \partial_z(\varphi_\eps \circ \xi_\eps ) \big) dv_g &= 2 \int_\Sigma g^{z \bar z} \hf g_{z \bar z} \sum_{i \in I} f_{i,d}^{zz} \prod_{j \neq i} \partial_{\eps_j}|_0 \big( \partial_z (\varphi_\eps \circ \xi_\eps) \big)^2 dv_g\,,
\end{align}
where we used $\varphi_\eps = \mathcal{O}(\eps)$, which implies that one of the derivatives must operate  on the $\partial_{\bar z}(\varphi_\eps \circ \xi_\eps)$, and then we used the $\eps$-derivative formulas \eqref{u_I} and \eqref{phi_I} for $\psi_\eps$ and $\varphi_\eps$, and $\partial_{\bar z} (\mathcal{G} f)^z =  \hf g_{z \bar z} f_{d}^{zz}$. %a minus sign comes from \psi^{-1} (and \varphi) and also from \gamma_{\bar z \bar z} \sim \mathcal{G} f^{zz}, so no overall sign change

For the $\partial_{\eps_I}$-derivative of the second term we use $f_{k,m}^{z \bar z} = 0$ and get
\begin{align}\label{qwe2}
\prod_{i \in I} \partial_{\eps_i}|_0 \int_\Sigma \sum_{k=1}^n \eps_k f_{k,m}^{\alpha \beta} \partial_\alpha (\varphi_\eps \circ \xi_\eps) \partial_\beta (\varphi_\eps \circ \xi_\eps) dv_g  &= \sum_{i \in I} \int f_{i,m}^{zz} \prod_{j \neq i} \partial_{\eps_j}|_0 \big( \partial_z (\varphi_\eps \circ \xi_\eps) \big)^2 dv_g\,.
\end{align}
By summing Equations \eqref{qwe1} and \eqref{qwe2} and using $f_{i,d} + f_{i,m} = f_i$, we get
\begin{align*}
\prod_{i \in I} \partial_{\eps_i}|_0 \int_\Sigma h_\eps^{\alpha \beta} \partial_\alpha (\varphi_\eps \circ \xi_\eps) \partial_\beta (\varphi_\eps \circ \xi_\eps) dv_g &=  \sum_{i \in I} \int_\Sigma f_i^{zz} \prod_{j \neq i} \partial_{\eps_j}|_0 \big( \partial_z (\varphi_{\hat \eps} \circ \xi_{\hat \eps}) \big)^2 dv_g\,,
\end{align*}
where $\hat \eps = (\eps_1,\hdots,\hat \eps_i, \hdots, \eps_n)$, where $\hat \eps_i$ means that $\eps_i$ is omitted. It is then easy to see that the above integral takes the form
\begin{align}\label{t2}
\int \prod_{j \in J} f_j^{z_j z_j}(z_j) H(\bz) dv_g(\bz)
\end{align}
for some function $H$ smooth in the region of non-coinciding points.

Next we evaluate the $\partial_{\eps_J}$-derivative in Equation \eqref{zxc}. From equations \eqref{u_I} and \eqref{Tm_orthog} it follows that
\begin{align}\label{t3}
&\partial_{\eps_J}|_0 \int_{\Sigma^{|K|}} \prod_{k \in K} f_{k,m}^{z_k z_k}(z_k) \langle \prod_{k \in K} T_m(z_k) \prod_j \Vg{g}{\alpha_j}(\psi_\eps(x_j)) \rangle_{g} dv_g(\bz) \\
&= \int_{\Sigma^{|K|+|J|}} \prod_{k \in K} f_{k,m}^{z_k z_k}(z_k) \prod_{j \in J} f_j^{z_j z_j}(z_j) D_J(\bx, z_J) \langle \prod_{k\in K}T_m(z_k) \prod_j \Vg{g}{\alpha_j}(x_j) \rangle dv_g(\bz) \nonumber \\
&= \int_{\Sigma^{|K|+|J|}} \prod_{k \in K} f_{k}^{z_k z_k}(z_k) \prod_{j \in J} f_j^{z_j z_j}(z_j) D_J(\bx, z_J) \langle \prod_{k\in K}T_m(z_k) \prod_j \Vg{g}{\alpha_j}(x_j) \rangle dv_g(\bz)\,, \nonumber 
\end{align}
where $D_J$ is a sum of terms of the form $F(\bx,z_J) \partial_{\bx}^\beta$, where $\partial_{\bx}^\beta = \partial_{x_i}^{\beta_1} \hdots \partial_{x_N}^{\beta_N}$ with $\sum_i |\beta_i| \leq |J|$, and $F$ is smooth in the region of non-coinciding points. From Proposition \ref{Tm_smooth} it follows that $\langle \prod_{k\in K}T_m(z_k) \prod_j \Vg{g}{\alpha_j}(x_j) \rangle $ is smooth on the whole domain of integration.

By combining Equations \eqref{t1}, \eqref{t2} and \eqref{t3}, the result follows.
\end{proof}

\begin{proposition}\label{T_transf}
Let $(x_1,\hdots,x_N) \in \Sigma^N$ and $(z_1,\hdots,z_n) \in (\Sigma \setminus \{x_j\}_j)^n$ be tuples of disjoint points. Then the stress-energy tensor satisfies
\begin{align}\label{t_diffeo}
\langle \prod_{k=1}^n T_{z_k z_k}(z_k) \prod_{j=1}^N \Vg{\psi^* g}{\alpha_j}(x_j) \rangle_{\psi^* g} &= \langle \prod_{k=1}^n (\psi^* T)_{z_k z_k}(z_k) \prod_{j=1}^N \Vg{g}{\alpha_j}(\psi(x_j)) \rangle_g\,,
\end{align}
where $(\psi^* T)_{zz} = (D^T\psi)_z^{\; \; \alpha} (T \circ \psi)_{\alpha \beta} (D \psi)^\beta_{\;\;z}$.

Denote $a_{\alpha \beta}(z,\varphi,g) = 4 \pi c \frac{\delta}{\delta g^{\alpha \beta}(z)} A(\varphi,g)$. Then
\begin{align}\label{t_weyl}
\langle \prod_{k=1}^n T_{z_k z_k}(z_k) \prod_{j=1}^N \Vg{e^\varphi g}{\alpha_j}(x_j) \rangle_{e^\varphi g} &= e^{cA(\varphi,g)- \sum_{j=1}^N \Delta_{\alpha_j} \varphi(x_j)} \langle \prod_{k=1}^n (T_{z_k z_k}(z_k) + a_{z_k z_k}(z_k,\varphi,g)) \prod_{j=1}^N \Vg{g}{\alpha_j}(x_j) \rangle_g\,.
\end{align}
\end{proposition}
\begin{proof}
The proof is almost identical as for $T_m$ in Proposition \ref{T_m_transformation}.
\end{proof}

\subsection{Conformal Ward identities}

Now we are ready to derive the conformal Ward identities.

\begin{theorem}\label{theorem}
Let $g$ be a hyperbolic metric and $f_k \in \sectt{tf}{}$ for $k=1,\hdots,n$ be symmetric tensors with mutually disjoint supports. Denote $\eps = (\eps_1,\hdots, \eps_n) \in \C^n$. Let $g_\eps$ be the perturbed metric
\begin{align*}
g_\eps^{zz} &= \overline{g_\eps^{\bar z \bar z}} =  \sum_{k=1}^n \eps_k f_k^{zz}\,, \\
g_\eps^{z \bar z} &= g_\eps^{\bar z z} = g^{z \bar z}\,.
\end{align*}
Then the function $\langle \prod_{k=1}^n T_{z_k z_k}(z_k) \prod_j \Vg{g}{\alpha_j}(x_j) \rangle_g$ defined by Equation \eqref{tpw} satisfies the conformal Ward identity \eqref{introward}.
\end{theorem}
\begin{proof}
Recall from Section \ref{beltrami_section} that we can write $g_\eps = e^{\varphi_\eps} \psi_\eps^* h_\eps$, where $\varphi_\eps \in C^\infty(\Sigma)$, $\psi_\eps \in \dif$ and $h_\eps \in \meth$. By Equation \eqref{tpw}
\begin{align*}
& \int_{\Sigma^n} \prod_{k=1}^n f_k^{z_k z_k}(z_k) \langle \prod_{k=1}^n T_{z_k z_k}(z_k) \prod_j \Vg{g}{\alpha_j}(x_j) \rangle_{g} dv_g(\bz) \\
&= 4 \pi \partial_{\eps_1}|_0 \int_{\Sigma^{n-1}} \prod_{k=2}^n f_k^{z_k z_k}(z_k) \langle \prod_{k=2}^n T_{z_k z_k}(z_k) \prod_j \Vg{g}{\alpha_j}(x_j) \rangle_{e^{\varphi_{\eps_1}} \psi_{\eps_1}^* h_{\eps_1}} dv_g(\bz)
\end{align*}
Next we use Proposition \ref{T_transf} to write the above expression as
\begin{align*}
& 4 \pi  \partial_{\eps_1}|_0 \int_{\Sigma^{n-1}} \prod_{k=2}^{n} f_k^{z_k z_k}(z_k) \Big( e^{c A(\varphi_{\eps_1},\psi_{\eps_1}^* h_{\eps_1}) - \sum_j \Delta_{\alpha_j} \varphi_{\eps_1}(x_j)} \\
& \qquad \qquad \qquad \qquad \qquad  \times  \langle \prod_{k=2}^{n} \big( (\psi_{\eps_1}^* T)_{z_k z_k}(z_k) + a(\psi_{\eps_1}(z_k), \varphi_{\eps_1}, h_{\eps_1}) \big) \prod_j \Vg{g_{\eps_1}}{\alpha_j}(\psi_{\eps_1}(x_j)) \rangle_{h_{\eps_1}} \Big) dv_g(\bz)
\end{align*}
\begin{align*}
&= 4 \pi c \int_{\Sigma^{n-1}} \prod_{k=2}^n f_k^{z_k z_k}(z_k) \partial_{\eps_1}|_0 A(\varphi_{\eps_1}, \psi_{\eps_1}^* h_{\eps_1}) \langle \prod_{k=2}^n T_{z_k z_k}(z_k) \prod_j \Vg{g}{\alpha_j}(x_j) \rangle_g dv_g(\bz) \\
& \quad - 4 \pi \sum_j \Delta_{\alpha_j} \int_{\Sigma^{n-1}} \prod_{k=2}^n f_k^{z_k z_k}(z_k) \partial_{\eps_1}|_0 \varphi_{\eps_1}(x_j) \langle \prod_{k=2}^n T_{z_k z_k}(z_k) \prod_j \Vg{g}{\alpha_j}(x_j) \rangle_g dv_g(\bz) \\
& \quad + 4 \pi \sum_{j=2}^n \int_{\Sigma^{n-1}}  \prod_{k=2}^n f_k^{z_k z_k}(z_k)  \partial_{\eps_1}|_0 a(\psi_{\eps_1}(z_k), \varphi_{\eps_1}, h_{\eps_1}) \langle \prod_{k \neq 1,j} T_{z_k z_k}(z_k) \prod_j \Vg{g}{\alpha_j}(x_j) \rangle_g dv_g(\bz) \\
& \quad +4 \pi  \int_{\Sigma^{n-1}}  \prod_{k=2}^n f_k^{z_k z_k}(z_k)  \partial_{\eps_1}|_0 \langle \prod_{k=2}^n (\psi_{\eps_1}^* T)_{z_k z_k}(z_k) \prod_j \Vg{g}{\alpha_j}(x_j) \rangle_g dv_g(\bz) \\
& \quad + 4 \pi \int_{\Sigma^{n-1}} \prod_{k=2}^n f_k^{z_k z_k}(z_k)  \partial_{\eps_1}|_0 \langle \prod_{k=2}^n T_{z_k z_k}(z_k) \prod_j \Vg{g}{\alpha_j}(\psi_{\eps_1}(x_j)) \rangle_g dv_g(\bz) \\
& \quad + 4 \pi \partial_{\eps_1}|_0  \int_{\Sigma^{n-1}} \prod_{k=2}^n f_k^{z_k z_k} (z_k) \langle \prod_{k=2}^n T_{z_k z_k}(z_k) \prod_j \Vg{h_{\eps_1}}{\alpha_j}(x_j) \rangle_{h_{\eps_1}} dv_g(\bz)\,.
\end{align*}
Next, we evaluate all the 6 terms above. Recall the definition $\eqref{anomaly}$ of $A$. We assumed that $g$ is hyperbolic, so $K_{\psi_{\eps_1}^* h_{\eps_1}} = K_g =  -2$ for all $\eps \geq 0$. Now from $\varphi_\eps = \mathcal{O}(\eps)$ and Equation \eqref{phi_I} it follows that
\begin{align*}
\partial_{\eps_1}|_0 A(\varphi_{\eps_1}, \psi_{\eps_1}^* h_{\eps_1}) &= \partial_{\eps_1}|_0 A(\varphi_{\eps_1},g) = \frac{1}{48 \pi} \int_\Sigma K_g(z) \partial_{\eps_1}|_0 \varphi_{\eps_1}(z) dv_g(z) = -\frac{1}{24\pi} \int_\Sigma \nabla_z u^z dv_g(z) = 0\,.
\end{align*}

We move on to the term $\partial_{\eps_1}|_0 a(\psi_{\eps_1}(z_k), \varphi_{\eps_1}, h_{\eps_1})$.  Let $g_\eps^{z\bar z} = g^{z \bar z}$ and $g_\eps^{z  z} = \eps \phi^{zz}$. By using definition of $a(z,\varphi,g)$ and Equations \eqref{K_var} and \eqref{phi_I}, we get
\begin{align*}
\int \phi^{zz}(z) \partial_{\eps_1}|_0 a(\psi_{\eps_1}(z),\varphi_{\eps_1},h_{\eps_1}) dv_g(z) &= 4 \pi c \partial_{\eps_1}|_0 \int  \phi^{zz}(z)  \frac{\delta}{\delta g^{zz}(z)} A(\varphi_{\eps_1},g) dv_g(z) \\
&= 4 \pi c \partial_{\eps_1}|_0 \partial_{\eps}|_0 \frac{1}{48\pi} \int K_{g_\eps} (z)\varphi_{\eps_1}(z) dv_g(z) \\
&= -\frac{c}{12} \int \phi^{zz}(z) \nabla_z^3 u_1^z(z) dv_g(z)\,,
\end{align*}
implying that $\partial_{\eps_1}|_0 a(\psi_{\eps_1}(z),\varphi_{\eps_1},h_{\eps_1}) = - \frac{c}{12} \nabla_z^3 u_1^z$.

From Equation \eqref{phi_I} we get $\partial_{\eps_1}|_0 \varphi_{\eps_1}(x_j) = - \nabla_{x_j} u_1^{x_j}(x_j)$, and it is straightforward to check, by using $\partial_{\eps_1}|_0 \overline{\psi_{\eps_1}} = 0$, that
\begin{align*}
\partial_{\eps_1}|_0 (\psi_{\eps_1}^* T)_{zz}(z) &= g_{z \bar z}^2 \partial_{\eps_1}|_0 \big( (D \psi_{\eps_1})^{\alpha \bar z} (T \circ \psi_{\eps_1})_{\alpha \beta} (D \psi_\eps)^{\beta \bar z} \big) \\
&= (2 \partial_z u_1^z + u_1^z \partial_z) T_{zz} \\
&= (2 \nabla_z u_1^z + u_1^z \nabla_z)T_{zz}\,,
\end{align*}
where we used $\nabla_z u_1^z = (\partial_z + \partial_z \sigma) u_1^z$ and $\nabla_z T_{zz} = (\partial_z - 2 \partial_z \sigma) T_{zz}$.

Finally,
\begin{align*}
\partial_{\eps_1}|_0 \Vg{g}{\alpha_j}(\psi_\eps(x_j)) &= u_1^z \nabla_z \Vg{g}{\alpha_j}(x_j)\,.
\end{align*}

By collecting all the terms together, we get

\begin{align*}
& \int \prod_{k=1}^n f_k^{z_k z_k}(z_k) \langle \prod_{k=1}^n T_{z_k z_k}(z_k) \prod_j \Vg{g}{\alpha_j}(x_j) \rangle_{g} dv_g(\bz) \\
&= - \frac{4 \pi c}{12} \sum_{i=2}^n  \int \prod_{k=1}^n f_k^{z_k z_k}(z_k)  \nabla_{z_i}^3 \mathcal{G}^{z_i}_{z_1 z_1} dv_g(z) \langle \prod_{k \neq 1,i} T_{z_k z_k}(z_k) \prod_j \Vg{g}{\alpha_j}(x_j) \rangle_g \\
&\quad +  4 \pi  \int \prod_{k=1}^n f_k^{z_k z_k}(z_k) \Big(  \sum_{i=2}^n \big( 2 \nabla_{z_i} \mathcal{G}^{z_i}_{z_1 z_1} + \mathcal{G}^{z_i}_{z_1 z_1} \nabla_{z_i} \big) + \sum_j ( \Delta_{\alpha_j} \nabla_{x_j} \mathcal{G}^{x_j}_{z_1 z_1} + \mathcal{G}^{x_j}_{z_1 z_1} \nabla_{x_j}) \Big) \\
& \qquad \qquad \qquad \qquad \qquad \times \langle \prod_{k=2}^{n} T_{z_k z_k}(z_k) \prod_j \Vg{g}{\alpha_j}(x_j) \rangle dv_g(z) \\
& \quad +4 \pi  \partial_{\eps_1}|_0 \int \prod_{k=2}^n f_k^{z_k z_k}(z_k) \langle \prod_{k=2}^n T_{z_k z_k}(z_k) \prod_j \Vg{g}{\alpha_j}(x_j) \rangle_{h_{\eps_1}} dv_g(\bz)\,.
\end{align*}
Denote $\tilde g_\eps^{zz} = \eps_1 f_{1,m}^{zz}+\sum_{k=2}^n \eps_k f_k^{zz}$. We write the last term on the right-hand side in the following way
\begin{align*}
&\partial_{\eps_1}|_0 \int \prod_{k=2}^n f_k^{z_k z_k}(z_k) \langle \prod_{k=2}^n T_{z_k z_k}(z_k) \prod_j \Vg{g}{\alpha_j}(x_j) \rangle_{h_{\eps_1}} dv_g(\bz) \\
&= \prod_{k=1}^n \partial_{\eps_k}|_0 \langle \prod_j \Vg{g}{\alpha_j}(x_j) \rangle_{\tilde g_\eps} \\
&= \prod_{k=2}^n \partial_{\eps_k}|_0 \int f_1^{z_1 z_1}(z_1) \langle T_m(z_1) \prod_j \Vg{g}{\alpha_j}(x_j) \rangle_{e^{\varphi_{\hat \eps}} \psi_{\hat \eps}^* h_{\hat \eps}} dv_g(z_1)\,,
\end{align*}
where $\hat \eps = (0,\eps_2,\hdots,\eps_n)$. Above we were free to change the order of the $\eps$-derivatives due to Proposition \ref{se_tensor}. Next we use the transformation properties from Proposition \ref{T_m_transformation} to get
\begin{align*}
&  \prod_{k=2}^n \partial_{\eps_k}|_0 \int f_1^{z_1 z_1}(z_1) \langle T_m(z_1) \prod_j \Vg{g}{\alpha_j}(x_j) \rangle_{e^{\varphi_{\eps}} \psi_{\eps}^* h_{\eps}} dv_g(z_1) \\
&= \prod_{k=2}^n \partial_{\eps_k}|_0 \int f_1^{z_1 z_1}(z_1) e^{c A(\varphi_\eps, \psi_\eps^* h_\eps) - \sum_j \Delta_{\alpha_j} \varphi_\eps(x_j)} \langle (\psi_\eps^* T_m)(z_1) \prod_j \Vg{g}{\alpha_j}(\psi_\eps(x_j)) \rangle_{h_\eps} dv_g(z_1)
\end{align*}
By separating the derivatives that operate on the $h_\eps$ from the rest, we get
\begin{align*}
& \prod_{k=2}^n \partial_{\eps_k}|_0 \int f_1^{z_1 z_1}(z_1) e^{c A(\varphi_\eps, \psi_\eps^* h_\eps) - \sum_j \Delta_{\alpha_j} \varphi_\eps(x_j)} \langle (\psi_\eps^* T_m)(z_1) \prod_j \Vg{g}{\alpha_j}(\psi_\eps(x_j)) \rangle_{h_\eps} dv_g(z_1) \\
&= \sum_{I \subset \{2,\hdots,n\}} \prod_{k \in I^c} \partial_{\eps_k}|_0 \int f_1^{z_1 z_1}(z_1) \prod_{k \in I} f_k^{z_k z_k}(z_k) e^{cA(\varphi_{\eps_{I^c}}, \psi_{\eps_{I^c}}^* g) - \sum_j \Delta_{\alpha_j} \varphi_{\eps_{I^c}}(x_j)} \\
& \qquad \qquad \qquad \qquad \times \langle (\psi^*_{\eps_{I^c}} T_m)(z_1) \prod_{k \in I} (\psi_{\eps_{I^c}}^* T_m)(z_k) \prod_j \Vg{g}{\alpha_j}(\psi_{\eps_{I^c}}(x_j)) \rangle_{g} dv_g(\bz) \\
&= \sum_{I \subset \{2,\hdots,n\}} \int \prod_{k \in I} f_k^{z_k z_k} (z_k) F_I(f_{I^c};\bx,\bz) D_{z}^{\alpha_I} D_x^{\beta_I} \langle T_m(z_1) \prod_{k \in I} T_m(z_k) \prod_j \Vg{g}{\alpha_j}(x_j) \rangle_g dv_g(\bz)\,,
\end{align*}
where $F_I(f_{I^c};\cdot,\cdot)$ is smooth in the region of non-coinciding points for each $I$ and $D_z^{\beta_I} = \prod_k \partial_{z_k}^{\beta_k}$. We denote the above distribution by $\langle T_m(z_1) \prod_{k=2}^n T_{z_k z_k}(z_k) \prod_j \Vg{g}{\alpha_j}(x_j) \rangle_g$. Now we have shown that
\begin{align*}
&\langle \prod_{k=1}^n T_{z_k z_k}(z_k) \prod_j \Vg{g}{\alpha_j}(x_j) \rangle_g \\
&= - \tfrac{4 \pi c}{12} \sum_{i=2}^n \nabla_{z_i}^3 \mathcal{G}^{z_i}_{z_1 z_1}  \langle \prod_{k \neq 1,i} T_{z_k z_k}(z_k) \prod_j \Vg{g}{\alpha_j}(x_j) \rangle_g \\
& \quad + 4\pi  \Big( \sum_{i=2}^n \big( 2 \nabla_{z_i} \mathcal{G}^{z_i}_{z_1 z_1} + \mathcal{G}^{z_i}_{z_1 z_1} \nabla_{z_i} \big) + \sum_j ( \Delta_{\alpha_j} \nabla_{x_j} \mathcal{G}^{x_j}_{z_1 z_1} + \mathcal{G}^{x_j}_{z_1 z_1} \nabla_{x_j}) \Big) \langle \prod_{k=2}^n T_{z_k z_k}(z_k) \prod_j \Vg{g}{\alpha_j}(x_j) \rangle_g \\
& \quad + \langle T_m(z_1) \prod_{k=2}^n T_{z_k z_k}(z_k) \prod_j \Vg{g}{\alpha_j}(x_j) \rangle_g\,,
\end{align*}
and it follows that $\langle \prod_{k=1}^n T_{z_k z_k}(z_k) \prod_j \Vg{g}{\alpha_j}(x_j) \rangle_g$ is smooth in $z_i$ and $x_j$ for non-coincident points.
\end{proof}

\appendix

\section{Distance function and Weyl transformations}\label{appendix}

Let $f \in \sectt{tf}{}$ and set $g_\eps^{zz} = \eps f^{zz}$, $g_\eps^{z \bar z} = g^{z \bar z}$. Fix a point $x \in \Sigma$ and a conformal coordinate $z$ defined on a chart $U$ around $x$. Now $g=e^\sigma |dz|^2$ on $U$ and it is possible to find a neighbourhood $U$ of $x$ and smooth maps $\psi_\eps: U \to \Sigma$ and $\varphi_\eps: U \to \R$ satisfying $g_\eps = e^{\varphi_\eps}\psi_\eps^* (e^\sigma |dz|^2)$. The function $\psi_\eps$ solves the Beltrami equation
\begin{equation}
\partial_{\bar z} \psi_\eps(z) = \mu_\eps(z) \partial_z \psi_\eps(z)
\end{equation}
with $\mu_\eps = \frac{\gamma_{\eps,\bar z \bar z}}{ \gamma_{ z \bar z} + \gamma_{\eps,z \bar z}}$ for $\gamma = \frac{g}{\sqrt{\det g}}$, $\gamma_\eps = \frac{g_\eps }{\sqrt{g_\eps}}$. From this it follows that (see Section \ref{beltrami_section} or Section 3.3. in \cite{KO})
\begin{align*}
\varphi_\eps = - \eps ( u^z \partial_z \sigma + \partial_z u^z) + \mathcal{O}(\eps^2)  = - \eps \nabla_z u^z + \mathcal{O}(\eps^2)\,,
\end{align*}
where $u^z = \partial_\eps|_0 \psi_\epsilon$.
For the distance function of $g_\eps$ we get
\begin{align*}
\ln d_{g_\eps}(x,y) &= \ln \Big( e^{\hf \varphi_\eps(x)} d_g(\psi_\eps(x),\psi_\eps(y))  + d_{g_\eps}(x,y) - e^{\hf \varphi_\eps(x)} d_g(\psi_\eps(x),\psi_\eps(y)) \Big) \\
&= \ln \big( e^{\hf \varphi_\eps(x)} d_g(\psi_\eps(x),\psi_\eps(y)) \big) + \frac{ d_{g_\eps}(x,y) - e^{\hf \varphi_\eps(x)} d_g(\psi_\eps(x),\psi_\eps(y))}{e^{\hf \varphi_\eps(x)} d_g(\psi_\eps(x),\psi_\eps(y))} + \hdots 
\end{align*}
where we used $\ln(x+y) = \ln x + \frac{y}{x} + \hdots$. By using the definition of the distance function, we get that
\begin{align*}
d_{g_\eps}(x,y) - e^{\hf \varphi_\eps(x)}  d_g(\psi_\eps(x),\psi_\eps(y)) &= \int_\gamma (\nabla_{\dot \gamma(t)} e^{\hf \varphi_\eps(\gamma(t))})d_g(x,\gamma(t)) \sqrt{(\psi^*_\eps g)_{\mu \nu}(\gamma(t)) \dot \gamma^\mu(t) \dot \gamma^\nu(t)} d t + \hdots 
\end{align*}
where we have Taylor expanded the function $e^{\hf \varphi_\eps}$ along the length minimizing geodesic $\gamma$. The absolute value of the $\partial_\eps|_0$ derivative of the above expression is bounded above by $C d_g(x,y)^2$, because the integrand is bounded from above by $C d_g(x,y)$, and the integration is over a path of length $d_g(x,y)$. It follows that
\begin{align}
\ln d_{g_\eps}(x,y) &= \hf \varphi_\eps(x) + \ln d_g(\psi_\eps(x),\psi_\eps(y)) + r_\eps(x,y)\,,
\end{align}
where $\partial_\eps^n|_0 r_\eps(x,y) = \mathcal{O}(d_g(x,y))$ for all $n \geq 0$.

\section{Relation to the Stress--Energy field}\label{appendix_b}

\noindent In this appendix we explain the equivalence of the two different definitions of the SE-tensor that have appeared in the literature. The rest of the article is completely independent of the discussion here, so some details are skipped. For simplicity, we assume that $\Sigma$ is the Riemann sphere endowed with the spherical metric
\begin{align}\label{spherical_metric}
g(z) &= \frac{4}{(1+|z|^2)^2}|dz|^2\,,
\end{align}
so that $K_g(z)=2$ for all $z$. We denote $e^{\sigma(z)} = \frac{4}{(1+|z|^2)^2}$. We want to comment on the relation between the $T_{zz}$-component of the Stress--Energy tensor and the \emph{Stress--Energy field} $\mathcal{T}$, which we formally\footnote{The expression has to be regularized for it to make rigorous sense, but to keep this appendix short we do not explicitly work with the regularizations.} define as
\begin{align}\label{se_field_ward}
\mathcal{T}(z) &= Q \partial_z^2 \big( X(z) + \tfrac{Q}{2} \sigma(z) \big) - \big(\partial_z (X(z)+\tfrac{Q}{2}\sigma(z) \big)^2 + \tfrac{1}{12} \big( \partial_z^2 \sigma(z) - \tfrac{1}{2} (\partial_z \sigma(z))^2 \big) +  \E_g[\big( \partial_z X(z) \big)^2]\,.
\end{align}
This definition agrees with the one used in \cite{ward} after one notices that
\begin{align}\label{t=0}
\partial_z^2 \sigma(z) - \tfrac{1}{2} (\partial_z \sigma(z))^2 =0
\end{align}
for the metric \eqref{spherical_metric}. The Stress--Energy field of Liouville theory was also studied in \cite{GKRV1, GKRV2}, where conformal Ward identities are derived in a regularized setting. Note that the above definition is written in a coordinate dependent way, and is not natural geometrically. First, we would like to write down the above formula in a more natural language. A simple computation shows that we have
\begin{align}\label{se_field_covariant}
\mathcal{T}(z) &= Q \nabla_z^2 X(z) - \big( \nabla_z X(z) \big)^2 +\E_g[\big( \nabla_z X(z) \big)^2] + \tfrac{c_L}{12} \big( \partial_z^2 \sigma(z) - \tfrac{1}{2} (\partial_z \sigma(z))^2  \big)\,,
\end{align}
where $c_L=1+6Q^2$ is the central charge of Liouville CFT. If the object $\sigma$ was a scalar ($0$-form), then the last term above would be a $2$-tensor, because then it would satisfy
\begin{align*}
t_{zz}(z) &:=  \partial_z^2 \sigma(z) - \tfrac{1}{2} (\partial_z \sigma(z))^2 = \nabla_z^2 \sigma(z) + \tfrac{1}{2} (\nabla_z \sigma(z))^2\,.
\end{align*}
In reality, $\sigma$ is not a scalar. A quick way to see this is to notice that $\partial_z \sigma = \Gamma^{z}_{zz}$ where $\Gamma$ is the Christoffel connection, so $\partial_z \sigma = \nabla_z \sigma$ does not transform as a $1$-form. It can also be checked that for a holomorphic coordinate change $z'=\psi(z)$ one has
\begin{align*}
t_{z'z'} (dz')^2 = t_{zz} dz^2 - \{\psi,z\}dz^2\,,
\end{align*}
where
\begin{align*}
\{f,z\} &:= \frac{\psi'''(z)}{\psi'(z)} - \frac{3}{2} \Big( \frac{\psi''(z)}{\psi'(z)} \Big)^2\,,
\end{align*}
is the Schwarzian derivative of $\psi$. This makes \eqref{se_field_covariant} look unnatural at first, but it turns out that $t_{zz}=0$ for the canonical metrics (spherical, flat, hyperbolic), and that one can always choose an atlas of coordinates consisting of isometries of the natural metrics, for which $\{\psi,z\} = 0$ always (for example the Möbius transforms have a vanishing Schwarzian derivative). Thus, we can discard $t_{zz}$ and we are then left with
\begin{align}
\tilde{\mathcal{T}}(z) &= \mathcal{T}(z) - \tfrac{c_L}{12} t_{zz}(z) = Q \nabla_z^2 X(z) - \big( \nabla_z X(z) \big)^2 +\E_g[\big( \nabla_z X(z) \big)^2]\,,
\end{align}
which is automatically a tensor. In what follows we will explain why the above expression agrees with $T_{zz}$.

Recall the definition of the Liouville expectation
\begin{align*}
\langle F \rangle_g &= Z_{\opn{GFF}}(\Sigma,g) \int_\R \E_g \big[ F(c+X) e^{- \frac{Q}{4\pi} \int_\Sigma(c+X)K_g dv_g - \mu e^{\gamma c} \mathbf{G}^\gamma_g(\Sigma)} \big] dc\,.
\end{align*}
The $g$-dependent terms are the partition function of the GFF $Z_{\opn{GFF}}(\Sigma,g)$, the expected value with respect to the GFF $\E_g$, the scalar curvature $K_g$ and the volume form $v_g$. The $g$-dependency of the GMC measure $\mathbf{G}_g^\gamma$ depends on $g$ only through the volume form $v_g$ and possibly the regularization procedure used to define it. These turn out not to be important, as we will soon see.

As we are only interested in the $T_{zz}$ component of the SE-tensor, we work with a perturbation $g_\eps$ of the metric $g$ that in terms of the inverse metric takes the form
\begin{align*}
g_\eps^{zz} &= \eps f^{zz}\,, \\
g_\eps^{\bar z \bar z} &= 0\,,\\
g_\eps^{z \bar z} &= g^{z \bar z}\,. \\
\end{align*}
On the Riemann sphere the $g$-dependency of $Z_{\opn{GFF}}(\Sigma,g)$ is fully described by the transformation laws
\begin{align*}
Z_{\opn{GFF}}(\Sigma,e^\omega g) &= e^{A(\omega,g)}Z_{\opn{GFF}}(\Sigma,g)\,, \\
Z_{\opn{GFF}}(\Sigma,\psi^* g) &= Z_{\opn{GFF}}(\Sigma,g)\,,
\end{align*}
where $A$ is the function defined in \eqref{anomaly}, because every metric on the sphere takes the form $e^\omega \psi^* g$, where $g$ is for example the spherical metric. To compute the derivative $\partial_\eps|_0 Z_{\opn{GFF}}(\Sigma,g_\eps)$ we recall Equation (3.47) in \cite{KO}, which states that
\begin{align*}
\partial_\eps|_0 A(\omega_\eps,\psi_\eps^*g ) &= - \tfrac{1}{48\pi}\int f^{zz} \big( \partial_z^2 \sigma(z) - \tfrac{1}{2} (\partial_z \sigma(z))^2 \big) dv_g(z) = 0\,,
\end{align*}
where in the latter equality we used \eqref{t=0}.

To compute the derivative $\partial_\eps| e^{- \frac{Q}{4\pi} \int_\Sigma (c+X)K_{g_\eps} dv_g}$ we invoke Lemma 4.2. in \cite{KO}, which states that 
\begin{align*}
\partial_\eps|_0 \int_\Sigma X(z) K_{g_\eps}(z) dv_g(z) &= \int_\Sigma f^{zz}(z) \big(-\partial_z^2 X(z) + \partial_z \sigma (z) X(z) \big) dv_g(z) = -\int_\Sigma f^{zz}(z) \nabla_z^2 X(z) dv_g(z)\,,
\end{align*}
which is just a restatement of $\partial_\eps|_0 K_{g_\eps}(z) = - \nabla_z^2 f^{zz}(z)$. This implies that
\begin{align}\label{b1}
\partial_\eps|_0 \int_\R \E_g[ F(c+X)e^{- \frac{Q}{4\pi} \int_\Sigma(c+X)K_{g_\eps} dv_g - \mu e^{\gamma c} \mathbf{G}_g^\gamma(\Sigma)}] &= \tfrac{Q}{4\pi} \int_\Sigma f^{zz}(z) \langle \nabla_z^2 X(z) F \rangle_g dv_g(z)\,.
\end{align}
Next, recall the formula \eqref{variation_of_gaussian}, which together with \eqref{deltag} and Gaussian integration by parts implies that
\begin{align*}
\partial_\eps|_0 \E_{g_\eps}[F(X)] &= \int_\Sigma f^{zz}(z) \E_g[\big( - \big( \nabla_z X(z) \big)^2 +\E_g[\big( \nabla_z X(z) \big)^2] \big) F(X)] dv_g(z)\,.
\end{align*}
It can be readily checked that for the type of perturbation $g_\eps$ we are considering,
\begin{align*}
dv_{g_\eps}(z) &= \sqrt{\det g_\eps (z)} |dz|^2 = \sqrt{\det g(z)} |dz|^2 + \mathcal{O}(\eps^2)\,,
\end{align*}
so that $\partial_\eps dv_{g_\eps}|_0=0$. Now, by combining the above observations (to be more precise, one should also work a little to see that there are no extra contributions from metric dependent regularizations), we get
\begin{align*}
\partial_\eps|_0 \langle F \rangle_{g_\eps} &= \tfrac{1}{4\pi} \int_\Sigma f^{zz}(z) \langle \mathcal{T}(z) F \rangle_g dv_g(z)\,,
\end{align*}
which explains the formula \eqref{se_field_ward}.

\end{document}